\documentclass[aps,pra,11pt,onecolumn,nofootinbib,tightenlines,superscriptaddress]{revtex4-2}

\usepackage{mathrsfs}
\usepackage{afterpage}
\usepackage{graphicx}
\usepackage{comment}
\usepackage{dsfont}
\usepackage{amsmath}
\usepackage{amsfonts}
\usepackage{amssymb}
\usepackage{braket}
\usepackage{amsthm}
\usepackage[colorlinks=true,urlcolor=blue,citecolor=blue,linkcolor=blue]{hyperref}
\usepackage{enumitem}

\usepackage{lipsum}
\usepackage{subcaption}
\usepackage{mwe}

\usepackage{sidecap}
\usepackage{natbib}
\usepackage{appendix}
\usepackage{multirow}
\usepackage{array}
\usepackage{physics}
\usepackage{mathtools}
\usepackage{xcolor}

\usepackage{tikz}
\usepackage{graphics}

\usepackage{caption}
\usepackage[paperwidth=225mm,paperheight=307mm,centering,hmargin=2.45cm,vmargin=2.5cm]{geometry}

\newtheorem{theorem}{Theorem}
\newtheorem{lemma}[theorem]{Lemma}
\newtheorem{proposition}[theorem]{Proposition}
\newtheorem{corollary}[theorem]{Corollary}
\newtheorem{definition}[theorem]{Definition}
\newtheorem{remark}{Remark}

\DeclareMathOperator{\supp}{supp}

\newcommand{\D}{\mathbb{D}}
\newcommand{\rleq}{\trianglelefteq}
\newcommand{\rgeq}{\trianglerighteq}
\newcommand{\lgeq}{\trianglelefteq}

\begin{document}

\title{\Large Maximal R\'enyi Relative Entropy for $\alpha>2$} 

\author{Roberto Rubboli}
\email{ror@math.ku.dk}
\affiliation{Department of Mathematical Sciences, University of Copenhagen, Universitetsparken 5, 2100 Denmark}

\begin{abstract}
Quantum relative entropies play a fundamental role in quantum information theory. In the classical setting, Rényi relative entropies constitute, up to linear combinations, the most general class of relative entropies, naturally motivating the search for their minimal and maximal quantum extensions. The minimal extension is known to be the reverse sandwiched Rényi relative entropy for \(\alpha\in[0,1/2)\) and the sandwiched Rényi relative entropy for \(\alpha\geq 1/2\). In contrast, the maximal extension had previously been identified only for \(\alpha\in[0,2]\), where it is given by the geometric Rényi relative entropy. In this work, we complete this characterization by proving that for \(\alpha>2\), the maximal extension is given by the \(\alpha\)-\(z\) Rényi relative entropy with \(z=\alpha-1\). As an application, we determine when an energy-incoherent state can be transformed into an energy-coherent state by a Gibbs-preserving operation assisted by an uncorrelated catalyst, thereby fully characterizing the coherence-generating power of this class of operations in the catalytic setting.
\end{abstract}

\maketitle

\tableofcontents

\section{Introduction}
Quantum relative entropies quantify the distinguishability between quantum states, playing a foundational role throughout quantum information theory~\cite{tomamichel2015quantum,khatri2020principles}. 
Yet, despite their importance, the structure of the class of all quantum relative entropies remains only partially understood.

In the classical setting, any relative entropy, that is, a functional on pairs of probability distributions that satisfies the data-processing inequality under stochastic maps and is additive under tensor products, can be written as a linear combination of the classical Rényi relative entropies~\cite{renyi1961divergence} (see~\cite{mu2021blackwell} and Section~\ref{sec: relative entropies} for a precise statement)
\begin{align}
     D_\alpha(p\| q)
=
\frac{1}{\alpha-1}
\log \sum_x p(x)^\alpha q(x)^{1-\alpha}.
 \end{align}
In the quantum setting, by contrast, a complete characterization remains open. This difficulty is reflected in the existence of several inequivalent quantum generalizations of the classical Rényi relative entropies.

However, several results are known about the minimal and maximal elements of this class. In particular, any quantum relative entropy \(\D_\alpha\) that reduces to the classical Rényi relative entropy \(D_\alpha\) on commuting states must satisfy the lower bound~\cite{hiai1991proper,mosonyi2015quantum,hayashi_2016-1,mosonyi2024geometric}
\begin{equation}
    \D_\alpha(\rho\|\sigma) \geq  \begin{cases}
    D_{\alpha,1-\alpha}(\rho\|\sigma), & \alpha \in [0,1/2]\\
         \widetilde{D}_{\alpha}(\rho\|\sigma), & \alpha \in (1/2,\infty]\,.
    \end{cases} 
\end{equation}
Here, $\widetilde{D}_{\alpha}$ is the sandwiched R\' enyi relative entropy~\cite{wilde2014strong,muller2013quantum}
\begin{align}
\label{eq:sandwiched-renyi}
    \widetilde{D}_{\alpha}(\rho\|\sigma) = \frac{1}{\alpha - 1} \log \Tr \left[ \left( \sigma^{\frac{1-\alpha}{2\alpha}} \rho \sigma^{\frac{1-\alpha}{2\alpha}} \right)^{\alpha} \right] \,,
\end{align}
while \(D_{\alpha,1-\alpha}\) is the reverse sandwiched relative entropy, which can be defined in terms of the sandwiched relative entropy with the arguments exchanged. Explicitly, 
\begin{align}
    D _ {\alpha,1-\alpha}(\rho\|\sigma)=\frac{\alpha}{1-\alpha}\, \widetilde{D}_{1-\alpha}(\sigma\|\rho) \,.
\end{align}
This quantity belongs to the family of \(\alpha\)-\(z\) Rényi relative entropies~\cite{audenaert13_alphaz} and is obtained by setting \(z=1-\alpha\).
Hence, the sandwiched and reverse sandwiched relative entropies correspond to the minimal quantum extension of the classical Rényi relative entropies.

Partial results concerning the maximal extension are also available. In particular, any quantum relative entropy \(\D_\alpha\) that reduces to the classical R\'enyi relative entropy \(D_\alpha\) on commuting states satisfies~\cite{matsumoto2013new}
\begin{equation}
    \label{eq: maximal gap}\D_\alpha(\rho\|\sigma) \leq  \begin{cases}
    \widehat{D}_{\alpha}(\rho\|\sigma), & \alpha \in [0,2]\\
         ? & \alpha \in (2,\infty]\,.
    \end{cases} 
\end{equation}
Here, $\widehat{D}_{\alpha}$ denotes the geometric relative entropy, defined as~\cite{matsumoto2013new,matsumoto2018maxdivergence}
\begin{align}
\label{def: geometric relative entropy}
    \widehat{D}_\alpha(\rho\|\sigma) = \frac{1}{\alpha-1}\log{\Tr\left[\sigma\left(\sigma^{-1/2}\rho\sigma^{-1/2}\right)^\alpha\right]} \,.
\end{align}
Notably, this quantity is defined through the trace of the \(\alpha\)-weighted geometric mean of the two matrices~\cite{kubo1980means}, from which many of its structural properties follow. Thus, in the range $\alpha\in[0,2]$, the geometric Rényi relative entropy constitutes the maximal quantum Rényi relative entropy. The characterization remains incomplete, however, as the maximal quantum extension for \(\alpha>2\) has yet to be identified.

\bigskip

In this work, we resolve this question and prove that for \(\alpha >2\), the maximal relative entropy is given by 
\begin{align}
\label{def: alpha-z for z=alpha-1}
D_{\alpha,\alpha-1}(\rho \| \sigma)=
\frac{1}{\alpha-1}\log{\Tr[\big(\rho^\frac{\alpha}{2(\alpha-1)}\sigma^{-1}\rho^\frac{\alpha}{2(\alpha-1)}\big)^{\alpha-1}]}\,.
\end{align}
In particular, our main result is the following.
\begin{theorem}
\label{thm: main theorem}
Let $\rho$ and $\sigma$ be quantum states. Then, every quantum relative entropy $\D_\alpha$ that reduces to the classical Rényi relative entropy $D_\alpha$ on classical states satisfies 
\begin{equation}
    \D_\alpha(\rho\|\sigma) \leq  \begin{cases}
    \widehat{D}_{\alpha}(\rho\|\sigma), & \alpha \in [0,2]\\
         D_{\alpha,\alpha-1}(\rho\|\sigma), & \alpha \in (2,\infty]\,.
    \end{cases} 
\end{equation}
\end{theorem}
This resolves a conjecture previously formulated by the author~\cite{rubboli2025thesis}.
The relative entropy $D_{\alpha,\alpha-1}$ belongs to the family of $\alpha$-$z$ Rényi relative entropies~\cite{audenaert13_alphaz} (see Section~\ref{sec: closed forms} for more details). Our result fills the remaining gap in~\eqref{eq: maximal gap}, yielding a complete characterization of the minimal and maximal quantum extensions of the classical Rényi relative entropies. The overall picture is illustrated in Figure~\ref{fig:extremal_entropies}. 
As a key intermediate step, we prove that, for every \(\alpha>2\), the regularized prepared Rényi relative entropy coincides with \(D_{\alpha,\alpha-1}\) (see Propositions~\ref{prop: converse} and~\ref{prop: achievability}). Combined with the previously known expression for \(\alpha\in[0,2]\), this yields a closed-form characterization valid for all \(\alpha\geq 0\), and provides a unified definition encompassing the geometric Rényi relative entropy for \(\alpha\in[0,2]\) and the \(\alpha\)-\(z\) Rényi relative entropy with \(z=\alpha-1\) for \(\alpha>2\).
To establish this result, we explicitly construct an asymptotically optimal preparation map that differs from the map introduced by Matsumoto~\cite{matsumoto2018maxdivergence}, whose optimality is restricted to the range \(\alpha\in[0,2]\) (see Remark~\ref{rem: optimal preparation map} for an explicit description of the preparation map).

\bigskip

\textbf{Applications.} As an application, we characterize the optimal asymptotic rate for the conversion of a pair of classical states \((p,q)\) into a pair of quantum states \((\rho,\sigma)\). More precisely, \(R((p,q)\to(\rho,\sigma))\) is the supremal rate \(R\) for which, for all sufficiently large \(n\), there exists a quantum channel \(\mathcal{E}_n\) implementing
\begin{align}
    p^{\otimes n}
    \xrightarrow{\ \mathcal{E}_n\ }
    \rho^{\otimes \lfloor Rn\rfloor},
    \quad
    q^{\otimes n}
    \xrightarrow{\ \mathcal{E}_n\ }
    \sigma^{\otimes \lfloor Rn\rfloor}.
\end{align}
We show that this rate is given by
\begin{equation}
    R\bigl((p,q)\to(\rho,\sigma)\bigr)
    =
    \min_{\D\in\mathcal{D}}
    \frac{\D(p\|q)}{\D(\rho\|\sigma)}\,,
\end{equation}
where \(\mathcal{D}\) consists of both argument orderings of
\(\widehat{D}_{\alpha}\) for \(\alpha\in[1/2,2]\) and of
\(D_{\alpha,\alpha-1}\) for \(\alpha\in[2,\infty]\) (see Corollary~\ref{cor: rate}).

Consequently, we derive necessary and sufficient conditions for transforming an energy-incoherent state into an energy-coherent one under catalytically assisted Gibbs-preserving operations. Unlike the more physically motivated class of thermal operations~\cite{janzing2000thermodynamic,horodecki2013fundamental,brandao2013resource}, Gibbs-preserving operations can generate coherence between distinct energy levels~\cite{faist2015gibbs}. Our result thus provides a complete characterization of their coherence-generating power in the catalytic setting.
Explicitly, in Theorem~\ref{thm:catalytic-gibbs-preserving} we prove that an
incoherent state \(p\) can be transformed into a coherent state \(\rho\)
with the assistance of a catalyst \(\nu\), returned unchanged and
uncorrelated with the system, if and only if
\begin{align}
    \D(p\|\gamma) \geq \D(\rho\|\gamma)
    \qquad \text{for all } \D\in\mathcal{D}.
\end{align}

\begin{figure}[htbp]
\centering
\begin{tikzpicture}

    \node at (0,0) {\includegraphics[width=0.6\textwidth]{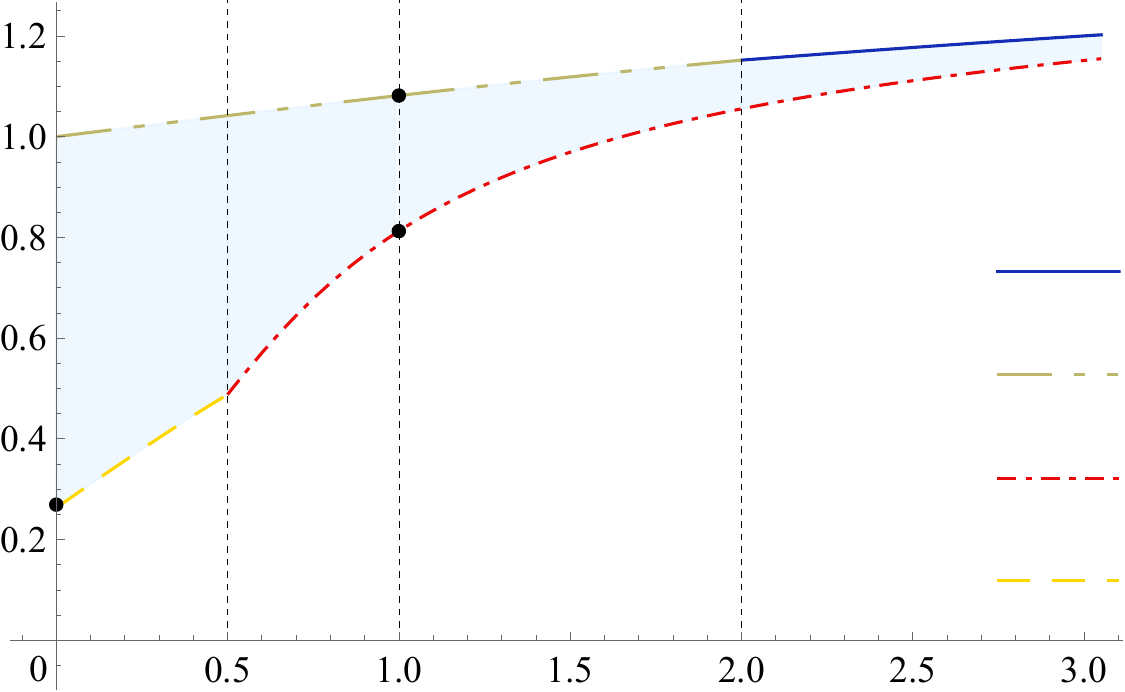}};

    \node at (-4.25,-1.75) {\large $D_{\min}$};
    \node at (-1.35,0.85) {\large $D$};
    \node at (-1.35,2.7) {\large $\widehat{D}$};
    \node at (4,3.15) {\large$D_{\max}$};

    \draw[->] (4.5,3.15) -- (5,3.18)
    node[midway, right, sloped] { $\;\,+\infty$};

    \node at (3.4,-2.2) {\large $D_{\alpha,1-\alpha}$};
    \node at (3.8,-0.25) {\large $\widehat{D}_{\alpha}$};
    \node at (3.8,-1.2) {\large $\widetilde{D}_\alpha$};
     \node at (3.4,0.7) {\large $D_{\alpha,\alpha-1}$};
    \node at (5.5,-3) {\large $\alpha$};

\end{tikzpicture}

\caption{
Minimal and maximal relative entropies evaluated for the states
\[
\rho =
\frac{1}{6}
\begin{pmatrix}
1 & 1 & 0 \\
1 & 3 & 2 \\
0 & 2 & 2
\end{pmatrix},
\qquad
\sigma =
\frac{1}{8}
\begin{pmatrix}
5 & 1 & 0 \\
1 & 2 & 1 \\
0 & 1 & 1
\end{pmatrix}.
\]
In this figure, all logarithms are taken to base \(2\).
For fixed $\alpha$, the minimal relative entropy is the reverse sandwiched relative entropy \(D_{\alpha,1-\alpha}\) for \(\alpha\in[0,1/2]\), and the sandwiched relative entropy for \(\alpha\ge 1/2\). The maximal relative entropy is the geometric relative entropy for \(\alpha\in[0,2]\). For \(\alpha>2\), this maximal quantity was previously unknown; we show here that it coincides with \(D_{\alpha,\alpha-1}\). All (normalized) quantum relative entropies are contained within the light blue region.  Across all \(\alpha\), the smallest and largest relative entropies are \(D_{\min}\) and \(D_{\max}\), respectively. $D$ and $\hat{D}$ represent the Umegaki and the Belavkin-Staszewski relative entropies, respectively.
}
\label{fig:extremal_entropies}
\end{figure}

\section{Preliminaries}
\subsection{Notation}
\label{sec: notation}
We denote by $\mathcal{P}(A)$ the set of positive semidefinite operators on a Hilbert space $A$. Moreover, we denote by $\mathcal{S}(A)$ the set of quantum states, i.e., the subset of $\mathcal{P}(A)$ with unit trace.
The Schatten $p$-norm is defined as $\|A\|_p = \big(\text{Tr}[(AA^\dagger)^\frac{p}{2} ]\big)^\frac{1}{p}$ for $p \geq 1$. For any Hermitian operator $L$, we denote by $\mathcal P_L$ the pinching map associated with $L$, defined for any linear operator $K$ as
\begin{align}
    \mathcal P_L(K)
    = \sum_{\lambda \in \operatorname{spec}(L)}
    P_{L,\lambda}\, K \, P_{L,\lambda},
\end{align}
where $P_{L,\lambda}$ denotes the orthogonal projector onto the eigenspace of $L$ corresponding to the eigenvalue $\lambda$. We denote by $|\mathrm{spec}(L)|$ the size of the spectrum of $L$, i.e., the number of distinct eigenvalues.
The pinching inequality states that for a positive operator $M$~\cite{hayashi2002optimal} (see also~\cite{tomamichel2015quantum})  
\begin{align}
\label{eq: pinching inequality}
\mathcal{P}_{L}(M) \geq \frac{M}{|\mathrm{spec}(L)|} 
\end{align}
We identify classical states with probability distributions by representing them as diagonal density operators. Accordingly, a pair of quantum states is called classical if the states commute, or equivalently, if they are simultaneously diagonalizable.

\subsection{Classical and quantum relative entropies}
\label{sec: relative entropies}

Let us first consider general functionals $\D$ of the form
\begin{align}
    \D : \ \bigcup_{A} \, \big\{ \mathcal{S}(A) \times \mathcal{S}(A) \big\} \to \mathbb R \cup \{+ \infty\} \,. \label{eq:D}
\end{align}
A functional $\D$
is called \textit{quantum relative entropy} if it satisfies the following properties~\cite{gour2020extensions}:
\begin{enumerate}
    \item \textbf{Data-processing inequality (DPI)}: For every quantum channel 
    $\mathcal{E}$, that is, every completely positive trace-preserving map,
    \begin{align}
        \D(\rho\| \sigma) \geq  \D(\mathcal{E}(\rho)\| \mathcal{E}(\sigma))\,.
    \end{align}
    \item \textbf{Additivity}:  $ \D(\rho_1\otimes \rho_2\| \sigma_1\otimes \sigma_2)=  \D(\rho_1\| \sigma_1)+ \D( \rho_2\| \sigma_2)$.
\end{enumerate}
An additional normalization condition is sometimes imposed. In particular, one often requires (see,
e.g.,~\cite{gour2020extensions})
\begin{align}
    \D(\ketbra{0}{0}\|I/2)=1.
\end{align}
In the present work, however, we do not impose this condition unless explicitly stated otherwise.

\bigskip

A classical relative entropy is a functional on pairs of probability distributions that is monotone under stochastic maps and additive under tensor products. Every such functional that is finite on pairs of probability distributions with coinciding supports can be represented, for any such pair \(p\) and \(q\), as a generalized linear combination of classical Rényi relative entropies~\cite{mu2021blackwell}. More precisely, it takes the form
\begin{align}
    \int_{\alpha\in[1/2,\infty]}\mathrm{d} \mu(\alpha)D_\alpha(p\|q) + \int_{\alpha\in[1/2,\infty]}\mathrm{d} \nu(\alpha)D_\alpha(q\|p),
\end{align}
where $\mu$ and $\nu$ are positive measures on the extended positive reals. For  $\alpha\in(0,1)\cup(1,\infty)$, the classical Rényi relative entropy is defined as~\cite{renyi1961divergence}
\begin{align}
    D_\alpha(p\| q) = \frac{1}{\alpha-1} \log \sum_x p(x)^\alpha q(x)^{1-\alpha}
\end{align}
Its values at $\alpha=0,1,\infty$ are defined by taking the corresponding limits.\footnote{The characterization above does not explicitly include \(D_{\min}\), as this relative entropy is trivial on pairs of states with coinciding supports. Pairs with non-coinciding supports require a more refined treatment. One possible approach is to extend the analysis directly to non-full-rank states, along the lines of~\cite{verhagen2025matrix}; another is to consider additional continuity assumptions, as in~\cite{gour2026quantum}. Finally, we note that the functionals $(p,q)\mapsto D_{\alpha}(q\|p)$ fail to satisfy the normalization axiom.} In particular, this result implies that the Rényi relative entropies constitute the extremal elements of the set of classical relative entropies, in the sense that every relative entropy can be expressed as a linear combination of them. When the probability distributions have different supports, the Rényi relative entropies are defined to be $+\infty$ when $p$ and $q$ are orthogonal for $\alpha  < 1$, or when the support of $q$ does not contain that of $p$ for $\alpha >1$. We denote the trace term by $Q_\alpha(p\|q)=\exp((\alpha-1)D_\alpha(p\|q))$.

\subsection{Measured and prepared relative entropies}
\label{sec: measured and prepared}

In the quantum setting, a natural approach to defining relative entropies is to extend the classical ones to quantum states via measurement or preparation maps. These operational constructions lift any classical relative entropy to a quantum divergence satisfying the DPI. Whenever the regularized limit exists, the resulting functional is weakly additive; in all cases considered here, it is fully additive and therefore defines a quantum relative entropy.
This framework traces back to the work of Donald~\cite{donald1986relent} and was further developed in~\cite{matsumoto2010relative,matsumoto2018maxdivergence}; see also~\cite{gour2020extensions} for a broader account of such extension methods.
Since, as discussed above, the Rényi relative entropies provide, through generalized linear combinations, the most general form of classical relative entropy, we henceforth focus on their measured and prepared extensions to the quantum setting.

The measured Rényi relative entropy is defined as~\cite{fuchs_1996,berta2017variational}
\begin{align}
    D^{\mathbb{M}}_\alpha(\rho\|\sigma) = \max_{\mathcal{M}} D_\alpha (\mathcal{M}(\rho) \| \mathcal{M}(\sigma) ),
\end{align}
where the maximization is taken over all measurements $\mathcal{M}$, i.e., quantum-to-classical channels. We denote the trace term by $Q^{\mathbb{M}}_\alpha(\rho\|\sigma)=\exp((\alpha-1)D^{\mathbb{M}}_\alpha(\rho\|\sigma))$.
The prepared relative entropy is defined as
\begin{align}
        D^{\mathbb{P}}_\alpha(\rho\|\sigma) = \inf_{p, q, \mathcal{F}} D_\alpha(p\|q)\, ,
    \end{align}
    where the infimum is taken over probability distributions $p, q$ of the same arbitrary finite dimension and preparation maps, i.e., classical-to-quantum channels, $\mathcal{F}$ such that $\mathcal{F}(p) = \rho$ and $\mathcal{F}(q) = \sigma$. We denote the trace term by $Q^{\mathbb{P}}_\alpha(\rho\|\sigma)=\exp((\alpha-1)D^{\mathbb{P}}_\alpha(\rho\|\sigma))$. At $\alpha=\infty$, we denote the corresponding prepared divergence by  $D_{\max}^{\mathbb{P}}$.
The DPI inequality applied to the measurement and preparation maps implies that any relative entropy $\D_\alpha$ that reduces to $D_{\alpha}$ on classical states satisfies 
    \begin{align}
        D_\alpha^{\mathbb{M}}(\rho\|\sigma)\leq \D_\alpha(\rho\|\sigma) \leq  D^{\mathbb{P}}_\alpha(\rho\|\sigma)\,.
    \end{align}
In general, however, the measured and prepared extensions need not be additive. One may therefore consider their regularizations in order to obtain additive quantities. In particular, since \(\D_\alpha\) is additive, it holds that, whenever the limit exists,
\begin{align}
\label{eq: sandwiched between regularized}
        \lim_{n\to\infty}\frac{1}{n}D_\alpha^{\mathbb{M}}(\rho^{\otimes n}\|\sigma^{\otimes n})\leq \D_\alpha(\rho\|\sigma) \leq  \lim_{n\to\infty}\frac{1}{n}D_\alpha^{\mathbb{P}}(\rho^{\otimes n}\|\sigma^{\otimes n})\,.
    \end{align}
The regularized quantities on the left and right satisfy the data-processing inequality and have the correct classical restriction by construction. Consequently, additivity is the only nontrivial property that remains to be established for them to define quantum relative entropies. Once this property is verified, the bounds above show that they constitute, respectively, the minimal and maximal quantum extensions of the classical R\'enyi relative entropy \(D_\alpha\).
As discussed in detail below, closed-form expressions are known for these quantities in all but one parameter regime, and their additivity follows readily from the corresponding formulas. The only remaining case is the maximal extension for \(\alpha>2\). One of the main results of this paper is a closed-form expression in this regime, which establishes the additivity of the corresponding regularized quantity.

\subsection{Closed-form expressions for the regularized measured and prepared relative entropies}
\label{sec: closed forms}
Closed-form expressions for the measured and prepared Rényi relative entropies are known in several ranges of the parameter \(\alpha\). In particular, the regularized measured Rényi relative entropy is given by~\cite{hiai1991proper,mosonyi2015quantum,hayashi_2016-1,mosonyi2024geometric}
\begin{align}
\label{eq: Reverse sandwiched from regularization}
     \lim_{n \rightarrow \infty} \frac{1}{n}D^{\mathbb{M}}_\alpha(\rho^{\otimes n}\|\sigma^{\otimes n}) = \begin{cases}
    D_{\alpha,1-\alpha}(\rho\|\sigma), & \alpha \in [0,1/2]\\
         \widetilde{D}_{\alpha}(\rho\|\sigma), & \alpha \in (1/2,\infty]\,.
    \end{cases} 
\end{align} 
Here, the sandwiched relative entropy is defined for $\alpha\in (0,1)\cup (1,\infty)$ as~\cite{muller2013quantum, wilde2014strong}
\begin{equation}
\label{definition sandwiched}
\widetilde{D}_{\alpha}(\rho \| \sigma) =
\frac{1}{\alpha-1}\log{\Tr\big[\big(\sigma^{\frac{1-\alpha}{2\alpha}}\rho \sigma^{\frac{1-\alpha}{2\alpha}}\big)^{\alpha}\big]} \,.
\end{equation}
If the states do not have full support,
the quantity evaluates to \( \infty \) if \( \rho \) and \( \sigma \) are orthogonal for \( \alpha < 1 \), or if the support of \( \sigma \) fails to contain that of \( \rho \) for $\alpha>1$.
The reverse sandwiched  is defined for $\alpha\in (0,1/2]$ as~\cite{audenaert13_alphaz}
\begin{equation}
    \label{def: reverse sandwiched}
 D _ {\alpha,1-\alpha}(\rho\|\sigma)=\frac{\alpha}{1-\alpha}\, \widetilde{D}_{1-\alpha}(\sigma\|\rho)=
\frac{1}{\alpha-1}\log{\Tr[\big(\rho^\frac{\alpha}{2(1-\alpha)}\sigma\rho^\frac{\alpha}{2(1-\alpha)}\big)^{1-\alpha}]}.
\end{equation}
If the states do not have full support,
the quantity evaluates to \( \infty \) if \( \rho \) and \( \sigma \) are orthogonal.

For $\alpha=0$, the reverse sandwiched is defined via the limit $\alpha\to 0$, where it converges to the 
min-relative entropy~\cite{renner2008security,datta2009min}
\begin{align}
    D_{\min}(\rho\|\sigma) &= -\log \Tr[ P_\rho \sigma ]  \,,
\end{align}
where $P_\rho$ is the projector onto the support of $\rho$. In the limit $\alpha\to \infty$, the sandwiched relative entropy converges to the max-relative entropy, given by~\cite{renner2008security,datta2009min}
\begin{align}
     D_{\max}(\rho\|\sigma) &= \inf \{ \lambda \in \mathbb{R} : \rho \leq \exp(\lambda) \sigma \} \,.
\end{align}
The reverse sandwiched relative entropies coincide with the \(\alpha\)-\(z\) Rényi relative entropies at \(z=1-\alpha\)~\cite{audenaert13_alphaz}. 
Notably, in the limit $\alpha\to 1$, the sandwiched relative entropy converges to the Umegaki relative entropy
\begin{align}
    D(\rho\|\sigma) = \Tr[\rho(\log{\rho}-\log{\sigma})] \,.
\end{align}

\bigskip 

For $\alpha\in[0,2]$, the prepared R\' enyi relative entropy is equal to 
\begin{align}
    D_\alpha^{\mathbb{P}}
    \left(
        \rho\|\sigma
    \right) = \widehat{D}_\alpha(\rho\|\sigma) \,.
\end{align}
Here, the geometric Rényi relative entropy is defined for $\alpha\in (0,1)\cup (1,2]$ as~\cite{matsumoto2018maxdivergence}
\begin{align}
    \widehat{D}_\alpha(\rho\|\sigma) = \frac{1}{\alpha-1}\log{\Tr\left[\sigma\left(\sigma^{-1/2}\rho\sigma^{-1/2}\right)^\alpha\right]} \,.
\end{align}
If the states do not have full support, for $\alpha <1 $, $\rho$ must be replaced by its absolutely continuous part with respect to $\sigma$. Moreover, the quantity is infinite if the absolutely continuous part of $\rho$ with respect to $\sigma$ vanishes for $\alpha<1$, or if the support of \( \sigma \) fails to contain that of \( \rho \) for $\alpha>1$.  
At $\alpha=0$ and $\alpha=1$, the geometric relative entropy is defined by the
corresponding limits. We refer to~\cite[Example III.36]{mosonyi2024geometric} for details.
For $\alpha\rightarrow 1$, it converges to the Belavkin--Staszewski relative entropy~\cite{belavkin1982c}
\begin{align}
    \widehat{D}(\rho\|\sigma) = \Tr\left[\rho \log{\left(\rho^{1/2}\sigma^{-1}\rho^{{1/2}}\right)}\right] .
\end{align}
It follows that, for \(\alpha \in [0,2]\), no regularization is required, since the prepared relative entropy admits an additive closed-form expression. In particular, for \(\alpha>2\), the geometric R\'enyi divergence no longer
provides the closed-form expression for the prepared divergence. This can also be understood through the notion of maximal
\(f\)-divergences~\cite{matsumoto2013new,hiai2017different}. The prepared
\(f\)-divergence, defined analogously to the prepared R\'enyi relative
entropy by replacing the R\'enyi relative entropy with the more general
\(f\)-divergence, admits an explicit closed-form expression when \(f\)
is operator convex. When \(f\) is chosen to be a power function, this
expression yields, up to a suitable rewriting, the geometric R\'enyi
divergences. However, the power function \(f(t)=t^\alpha\) ceases to be
operator convex for \(\alpha>2\).

We now introduce our main quantity of interest, which we later show to correspond to the regularized prepared relative entropy for $\alpha>2$. The \(\alpha\)-\(z\) Rényi relative entropy with \(z=\alpha-1\) is defined for \(\alpha \in [2,\infty)\) as~\cite{audenaert13_alphaz}
\begin{equation}
    \label{def: Dalphaalpha-1 sandwiched}
 D_{\alpha,\alpha-1}(\rho \| \sigma)=
\frac{1}{\alpha-1}\log{\Tr[\big(\rho^\frac{\alpha}{2(\alpha-1)}\sigma^{-1}\rho^\frac{\alpha}{2(\alpha-1)}\big)^{\alpha-1}]}\,.
\end{equation}
If the states do not have full support,
the quantity evaluates to \( \infty \) if the support of \( \sigma \) fails to contain that of \( \rho \). For $\alpha =\infty$, we define this quantity via the corresponding limit, where it converges to the max-relative entropy (see, e.g.,~\cite[Lemma 23]{rubboli2024mixed}). Note that for $\alpha=2$ it equals both the geometric relative entropy and the Petz relative entropy~\cite{petz1986quasi}.
We also denote the trace term by $Q_{\alpha,\alpha-1}(\rho\|\sigma)=\exp((\alpha-1)D_{\alpha,\alpha-1}(\rho\|\sigma))$.
The quantity $D_{\alpha,\alpha-1}$ satisfies the data-processing inequality if and only if $\alpha \geq 2$~\cite{zhang2020wigner}. It coincides with the $\alpha$-$z$ Rényi relative entropies with $z=\alpha-1$~\cite{audenaert13_alphaz}.

Finally, we mention that across the entire range of $\alpha$, the min- and max-relative entropies are, respectively, the smallest and largest normalized quantum relative entropies; see, e.g.,~\cite{gour2020extensions}. Hence, for any such relative entropy \(\D\),
\begin{align}
D_{\min}(\rho\|\sigma)
\leq
\D(\rho\|\sigma)
\leq
D_{\max}(\rho\|\sigma).
\end{align}
The resulting hierarchy is illustrated in Fig.~\ref{fig:extremal_entropies}.

\section{Converse bound}
The proof of the main result in Theorem~\ref{thm: main theorem} is divided into two parts. First, we establish a lower bound on the regularized prepared R\'enyi relative entropy. In the subsequent section, we derive the matching upper bound, which constitutes the achievability part. Together, these bounds show that the regularized prepared R\'enyi relative entropy coincides with \(D_{\alpha,\alpha-1}\). Since the latter quantity is additive under tensor products, the discussion in Section~\ref{sec: relative entropies} then implies our main result, namely, that every quantum relative entropy extending the classical R\'enyi relative entropy \(D_\alpha\) is upper bounded by \(D_{\alpha,\alpha-1}\) for \(\alpha>2\).

The proof of the lower bound follows directly from the data-processing inequality for \(D_{\alpha,\alpha-1}\), which holds for \(\alpha>2\)~\cite{zhang2020wigner}.
\begin{proposition}
\label{prop: converse}
Let $\rho$ and $\sigma$ be quantum states. Then, for any $\alpha\geq 2$, it holds that
\begin{align}
    \lim_{n\to\infty}
    \frac1n
    D_\alpha^{\mathbb{P}}
    \left(
        \rho^{\otimes n}\|\sigma^{\otimes n}
    \right)
    \geq 
    D_{\alpha,\alpha-1}(\rho\|\sigma)
\end{align}
\end{proposition}
\begin{proof}

Let us consider first the case where $\operatorname{supp}(\rho)\nsubseteq \operatorname{supp}(\sigma)$.
Then there exists a vector \(|v\rangle\) such that $\langle v|\sigma|v\rangle=0$ and $
    \langle v|\rho|v\rangle>0 $. A preparation map for \((\rho,\sigma)\) consists of two probability
distributions \(p,q\) and a classical-to-quantum
channel \(\mathcal{F}\) such that $\mathcal{F}(p)=\rho$ and $\mathcal{F}(q)=\sigma$.
Since \(\mathcal{F}\) is a classical-to-quantum channel, there exist density operators
\(\tau_x=\mathcal{F}(e_x)\) such that, by linearity,
\begin{align}
    \rho=\sum_x p_x\tau_x,
    \qquad
    \sigma=\sum_x q_x\tau_x .
\end{align}
Since $0=\langle v|\sigma|v\rangle
      =\sum_x q_x \langle v|\tau_x|v\rangle$
and all terms in the sum are nonnegative, we must have
\begin{align}
    q_x>0
    \quad\Longrightarrow\quad
    \langle v|\tau_x|v\rangle=0 .
\end{align}
On the other hand, $0<\langle v|\rho|v\rangle
      =\sum_x p_x \langle v|\tau_x|v\rangle$.
Hence there must exist some \(x\) such that
\begin{align}
    p_x>0,
    \qquad
    \langle v|\tau_x|v\rangle>0 .
\end{align}
For this \(x\), necessarily \(q_x=0\). Therefore, every preparation map contains an index with $p_x>0$ and $ q_x=0$.
Since \(\alpha>1\), the classical Rényi relative entropy is infinite. Since this holds for all preparation maps, the prepared relative entropy is infinite.

Now consider the case where $\operatorname{supp}(\rho)\subseteq \operatorname{supp}(\sigma)$.
The lower bound follows from the data-processing inequality satisfied by $D_{\alpha,\alpha-1}$ for $\alpha\geq 2$. Let us consider a preparation map $\mathcal{F}_n$ and probability distributions $r_n$ and $s_n$ such that
\begin{align}
    \mathcal{F}_n(r_n)=\rho^{\otimes n},\quad \mathcal{F}_n(s_n)=
    \sigma^{\otimes n} \,.
\end{align}
By the data-processing inequality,
\begin{align}
    n D_{\alpha,\alpha-1}
    (
        \rho\| \sigma) = D_{\alpha,\alpha-1}
    \left(
        \rho^{\otimes n}\|\sigma^{\otimes n}
    \right)
    \le
    D_{\alpha,\alpha-1}(r_n\| s_n) = D_{\alpha}(r_n\| s_n) .
\end{align}
Taking the infimum over all preparation maps and distributions gives
\begin{align}
    n D_{\alpha,\alpha-1}(
        \rho\| \sigma)
    \le
    D_\alpha^{\mathbb{P}}
    \left(
        \rho^{\otimes n}\|\sigma^{\otimes n}
    \right).
\end{align}
Multiplying by $1/n$ and taking the limit $n\to \infty$ proves the lower bound. 
\end{proof}

\section{Achievability}
In this section, we derive an upper bound on the regularized prepared R\'enyi relative entropy in terms of \(D_{\alpha,\alpha-1}\). In particular, we construct an explicit preparation map that upper bounds the prepared R\'enyi relative entropy by \(D_{\alpha,\alpha-1}\), up to an additive correction given by the logarithm of the size of the spectrum of \(\rho\). When applied to tensor powers, this correction grows only logarithmically in the number of copies and therefore vanishes upon regularization, yielding the desired upper bound.

We begin by expressing \(Q_\alpha^{\mathbb{P}}(\rho\|\sigma)\) as a variational problem over positive semidefinite operators. Since \(Q_\alpha^{\mathbb{P}}(\rho\|\sigma)=\exp((\alpha-1)D_\alpha^{\mathbb{P}}(\rho\|\sigma))\), this formulation equivalently provides a variational characterization of the prepared R\'enyi relative entropy. This is the specialization of Matsumoto's reverse-test formulation of maximal \(f\)-divergences to the power function \(f_\alpha(t)=t^\alpha\)~\cite[Section~4.2, Eq.~(4.2)]{matsumoto2013new}. 

\begin{lemma}
\label{lem: optimization problem}
Let \(\rho,\sigma\) be quantum states such that \(\supp(\rho) \subseteq\supp (\sigma)\).
For any $\alpha \in (1,\infty)$, it holds that
\begin{align}
    Q_\alpha^{\mathbb{P}}(\rho\| \sigma)
    =
    \inf_{\substack{X_x\ge0,\ t_x\ge0\\
    \sum_x X_x=\sigma\\
    \sum_x t_xX_x=\rho}}
    \sum_x t_x^\alpha \Tr X_x.
\end{align}
\end{lemma}

\begin{proof}
We first show that the left-hand side is larger than the right-hand side.
Let \((p,q,\mathcal{F})\) be a preparation map. Since
\(\mathcal{F}\) is a classical-to-quantum channel, its action is determined by its values on the classical basis vectors. We denote these output states by
\begin{align}
    \tau_x = \mathcal{F}(e_x) \,,
\end{align}
where \(e_x\) is the classical probability vector with a one in position \(x\) and zeros elsewhere. We then have that
\begin{align}
     \rho=\mathcal{F}(p)=\sum_x p_x\tau_x,
    \qquad
    \sigma=\mathcal{F}(q)=\sum_x q_x\tau_x.
\end{align}
Define $X_x=q_x\tau_x$.
If \(q_x>0\), define $ t_x=\frac{p_x}{q_x}$.
Terms with \(q_x=0\) and \(p_x=0\) may be discarded. If \(q_x=0\) but
\(p_x>0\), then the classical quantity
\(Q_\alpha(p\| q)\) is infinite for \(\alpha>1\), so such terms are
irrelevant for the infimum. We then have that
\begin{align}
    \sum_xX_x=\sigma,
    \qquad
    \sum_xt_xX_x=\rho.
\end{align}
Moreover,
\begin{align}
    Q_\alpha(p\| q)
    =
    \sum_x p_x^\alpha q_x^{1-\alpha}
    =
    \sum_x
    \left(\frac{p_x}{q_x}\right)^\alpha
    q_x
    =
    \sum_x t_x^\alpha \Tr X_x.
\end{align}
Thus, every preparation map gives a feasible point of the optimization problem, and hence the inequality holds.

Let us now prove the converse, namely that the left-hand side is less than the right-hand side. Suppose \(X_x,t_x\) are feasible. Define
\begin{align}
    q_x=\Tr X_x,
    \qquad
    p_x=t_x\Tr X_x,
\end{align}
and, when \(\Tr X_x>0\),
\begin{align}
     \tau_x=\frac{X_x}{\Tr X_x}.
\end{align}
Zero-trace terms may be discarded. Since \(\rho,\sigma\) are states,
\begin{align}
     \sum_xq_x=\Tr \sigma=1,
    \qquad
    \sum_xp_x=\Tr \rho=1.
\end{align}
Thus \(p,q\) are probability distributions. The preparation channel
\(\mathcal{F}(e_x)=\tau_x\) satisfies
\begin{align}
     &\mathcal{F}(q)=\sum_xq_x\tau_x=\sum_xX_x=\sigma,\\
     & \mathcal{F}(p)=\sum_xp_x\tau_x
    =
    \sum_xt_xX_x=\rho.
\end{align}
The classical cost is
\begin{align}
    Q_\alpha(p\| q)
    =
    \sum_xp_x^\alpha q_x^{1-\alpha}
    =
    \sum_xt_x^\alpha\Tr X_x.
\end{align}
Therefore, we obtain the desired inequality.
\end{proof}

We next derive an upper bound on the prepared R\'enyi relative entropy in terms of \(D_{\alpha,\alpha-1}\) by constructing an explicit preparation map.
\begin{lemma}
\label{lem: upper bound}
Let \(\rho,\sigma\) be quantum states. Then, for every \(\alpha \in [2,\infty)\),
\begin{align}
     D_\alpha^{\mathbb{P}}(\rho\| \sigma)
    \le
    D_{\alpha,\alpha-1}(\rho\| \sigma)+\log  |\textup{spec}(\rho)|.
\end{align}
In addition, $ D_{\max}^{\mathbb{P}}(\rho\| \sigma)=  D_{\max}(\rho\| \sigma)$.
\end{lemma}

\begin{proof}
In the case $\operatorname{supp}(\rho)\nsubseteq \operatorname{supp}(\sigma)$ an analogous argument to the one in the proof of Proposition~\ref{prop: converse} shows that both quantities are infinite. Let us consider now the case $\operatorname{supp}(\rho)\subseteq \operatorname{supp}(\sigma)$.
If \(\supp(\rho)\subseteq\supp(\sigma)\) but \(\sigma\) is not full-rank, we restrict the whole argument to \(\supp(\sigma)\). On this
subspace, \(\sigma\) is strictly positive and all inverses below are ordinary
inverses.

We equivalently prove that
\begin{align}
    Q_\alpha^{\mathbb{P}}(\rho\| \sigma)
    \le
    |\textup{spec}(\rho)|^{\alpha-1}
    Q_{\alpha,\alpha-1}(\rho\| \sigma).
\end{align}
Let us denote by $m=|\text{spec}(\rho)|$ and $P_j$ the orthogonal projectors onto the eigenspace of $\rho$ corresponding to its distinct $j$-th eigenvalue $r_j$. We denote by \(\mathcal P\) the pinching map with respect to \(\rho\), namely 
\begin{align}
    \mathcal{P}(X)=\sum_{j=1}^m P_jX P_j.
\end{align}
By the pinching inequality~\eqref{eq: pinching inequality} applied to $\sigma^{-1}$, we obtain that $\sigma^{-1}\le m\,\mathcal{P}(\sigma^{-1})$.
Since the inverse is operator antimonotone, we obtain that
\begin{align}
    \bigl(m\,\mathcal{P}(\sigma^{-1})\bigr)^{-1}\le \sigma.
\end{align}
Define
\begin{align}
    \tilde{\sigma}=\bigl(m\,\mathcal{P}(\sigma^{-1})\bigr)^{-1}.
\end{align}
The above bounds imply that $0\le \tilde{\sigma}\le \sigma$.

Because \(\tilde{\sigma}\) is block diagonal with respect to the projections \(P_j\), write
\begin{align}
    \tilde{\sigma}=\bigoplus_{j=1}^m \tilde{\sigma}_j,
    \qquad
    \tilde{\sigma}_j=P_j \tilde{\sigma} P_j.
\end{align}
Also, define $\hat{\sigma}_j= P_j \sigma^{-1}P_j$. Then $\tilde{\sigma}_j
    =
    \frac1m \hat{\sigma}_j^{-1}$, where the inverse is understood to be taken on the support of \(\widehat{\sigma}_j\). We then diagonalize each \(\tilde{\sigma}_j\)
    \begin{align}
         \tilde{\sigma}_j=\sum_k c_{jk}Q_{jk},
    \end{align}
where $c_{jk}>0$, $Q_{jk}\le P_j$ and $\sum_kQ_{jk}=P_j$. We then construct a feasible point for the optimization problem of the prepared relative entropy in Lemma~\ref{lem: optimization problem}  as follows. We set the index $x=(j,k)$ and
\begin{align}
    X_{jk}=c_{jk}Q_{jk},
    \qquad
    t_{jk}=\frac{r_j}{c_{jk}},
\end{align}
where $r_j$ is the $j$-th eigenvalue of $\rho$, 
and we add one residual outcome
\begin{align}
    X_0=\sigma-\tilde{\sigma},
    \qquad
    t_0=0.
\end{align}
Since $\sigma \geq \tilde{\sigma}$, we have that $X_0\ge0$. 

We now check the constraints to show that it is a feasible point.
First,
\begin{align}
    \sum_{j,k}X_{jk}+X_0
    =
    \sum_{j,k}c_{jk}Q_{jk}+\sigma-\tilde{\sigma}
    =
    \tilde{\sigma}+\sigma-\tilde{\sigma}
    =
    \sigma.
\end{align}
Second,
\begin{align}
    \sum_{j,k}t_{jk}X_{jk}+t_0X_0
    =
    \sum_{j,k}
    \frac{r_j}{c_{jk}}
    c_{jk}Q_{jk}
    +
    0\cdot(\sigma-\tilde{\sigma}) = \sum_j r_j\sum_kQ_{jk} = \sum_j r_jP_j
    =
    \rho.
\end{align}
Hence, the constructed variables are feasible for the optimization problem.
The cost of this feasible point is
\begin{align}
    \sum_{j,k} t_{jk}^{\alpha}\Tr X_{jk}
    +
    t_0^\alpha\Tr X_0=
    \sum_{j,k} t_{jk}^{\alpha}\Tr X_{jk}
    =
    \sum_{j,k}
    \left(\frac{r_j}{c_{jk}}\right)^\alpha
    \Tr(c_{jk}Q_{jk})
    =
    \sum_j r_j^\alpha \Tr \tilde{\sigma}_j^{1-\alpha}.
\end{align}
Using \(\tilde{\sigma}_j=\frac1m\hat{\sigma}_j^{-1}\), we get
\begin{align}
     \tilde{\sigma}_j^{1-\alpha}
    =
    \left(\frac1m\hat{\sigma}_j^{-1}\right)^{1-\alpha}
    =
    m^{\alpha-1}\hat{\sigma}_j^{\alpha-1}.
\end{align}
Therefore, by feasibility
\begin{align}
    \label{eq: feasibility prepared}Q_\alpha^{\mathbb{P}}(\rho\| \sigma)
    \le
    m^{\alpha-1}
    \sum_j r_j^\alpha
    \Tr \hat{\sigma}_j^{\alpha-1}.
\end{align}
It remains to compare the above sum with \(Q_{\alpha,\alpha-1}(\rho\| \sigma)\).

Since $\alpha-1 \geq 1$, Schatten contractivity under pinching in Lemma~\ref{lem: Schatten contractivity} gives
\begin{align}
Q_{\alpha,\alpha-1}(\rho \| \sigma) \geq \Tr[\Big(\mathcal{P}\big(\rho^\frac{\alpha}{2(\alpha-1)}\sigma^{-1}\rho^\frac{\alpha}{2(\alpha-1)}\big)\Big)^{\alpha-1}]\,.
\end{align}
Now
\begin{align}
    P_j \rho^\frac{\alpha}{2(\alpha-1)}\sigma^{-1}\rho^\frac{\alpha}{2(\alpha-1)} P_j= r_j^\frac{\alpha}{\alpha-1} P_j\sigma^{-1}P_j
    =
    r_j^\frac{\alpha}{\alpha-1} \hat{\sigma}_j.
\end{align}
Hence, $\mathcal{P}(\rho^\frac{\alpha}{2(\alpha-1)}\sigma^{-1}\rho^\frac{\alpha}{2(\alpha-1)})
    =
    \bigoplus_j r_j^\frac{\alpha}{\alpha-1} \hat{\sigma}_j$.
Thus,
\begin{align}
    \Tr[\Big(\mathcal{P}\big(\rho^\frac{\alpha}{2(\alpha-1)}\sigma^{-1}\rho^\frac{\alpha}{2(\alpha-1)}\big)\Big)^{\alpha-1}]=  \sum_j
    \Tr
    [(r_j^\frac{\alpha}{\alpha-1}\hat{\sigma}_j)^{\alpha-1}] 
    =
    \sum_j
    r_j^\alpha\Tr \hat{\sigma}_j^{\alpha-1}.
\end{align}
Combining the previous inequalities, we obtain that
\begin{align}
\label{eq: lower bound Q}
    Q_{\alpha,\alpha-1}(\rho\| \sigma)
    \ge
    \sum_j
    r_j^\alpha\Tr \hat{\sigma}_j^{\alpha-1}.
\end{align}
Finally,~\eqref{eq: feasibility prepared} and~\eqref{eq: lower bound Q} imply
\begin{align}
    Q_\alpha^{\mathbb{P}}(\rho\| \sigma)
    \le
    m^{\alpha-1}
    Q_{\alpha,\alpha-1}(\rho\| \sigma).
\end{align}

Finally, let us consider the case $\alpha=\infty$. Set $c= \exp\bigl(D_{\max}(\rho\|\sigma)\bigr)= \min\bigl\{c'>0:\rho\leq c'\sigma\bigr\}>1$. Consider the distributions
\begin{align}
    p
    =
    (1,0),
    \quad
    q
    &=
    \left(\frac{1}{c},1-\frac{1}{c}\right),
\end{align}
and the preparation states
\begin{align}
    \tau_1
    =
    \rho,
    \quad
    \tau_2
    =
    \frac{c\sigma-\rho}{c-1}.
\end{align}
The operator \(\tau_2\) is a state, since \(c\sigma-\rho\geq0\) and
\(\Tr[c\sigma-\rho]=c-1\). Moreover,
\begin{align}
    \mathcal{F}(p)
    =
    \tau_1
    =
    \rho,
    \quad 
    \mathcal{F}(q)
    =
    \frac{1}{c}\tau_1
    +
    \left(1-\frac{1}{c}\right)\tau_2
    =
    \sigma,
\end{align}
and $ D_{\max}(p\|q)=\log{c}$.
Together with the data-processing inequality, this proves the 
identity
\begin{align}
    D_{\max}^{\mathbb{P}}(\rho\|\sigma)
    =
    D_{\max}(\rho\|\sigma).
\end{align}
\end{proof}

\begin{remark}
\label{rem: optimal preparation map}
The above result is constructive: it yields an explicit preparation map for $\rho$ and $\sigma$ whose cost matches $D_{\alpha,\alpha-1}$ up to a spectral-size correction. As shown below, this correction vanishes asymptotically, rendering the preparation map asymptotically optimal.
The map is constructed as follows.

Let $r_j$ be the eigenvalues of $\rho$ with corresponding eigenprojectors $P_j$, and define the operators
\begin{align}
\widetilde{\sigma} = \frac{\mathcal{P}_\rho(\sigma^{-1})^{-1}}{|\textup{spec}(\rho)|}, \qquad \widetilde{\sigma}_j = P_j\widetilde{\sigma}P_j.
\end{align}
We diagonalize each block $\widetilde{\sigma}_j$ on its support as
\begin{align}
\widetilde{\sigma}_j = \sum_{k=1}^{k_{\max}(j)} c_{jk} Q_{jk},
\end{align}
where $Q_{jk}$ are the spectral projectors that satisfy $\sum_{k=1}^{k_{\max}(j)} Q_{jk} = P_j$, and $k_{\max}(j)$ is the spectral size of $\widetilde{\sigma}_j$ on its support.

The asymptotically optimal preparation map $F_{\rho,\sigma}$ is a classical-to-quantum channel over the alphabet $\mathcal{X} = \{(j,k) : 1 \leq j \leq m,\ 1 \leq k \leq k_{\max}(j)\} \cup \{0\}$. Its action on the standard basis vectors is given by $F_{\rho,\sigma}(e_{jk}) = \tau_{jk}$ and $F_{\rho,\sigma}(e_0) = \tau_0$, preparing the respective states
\begin{align}
\tau_{jk} = \frac{Q_{jk}}{\operatorname{Tr}[Q_{jk}]}, \qquad \tau_0 = \frac{\sigma-\widetilde{\sigma}}{\operatorname{Tr}[\sigma-\widetilde{\sigma}]}.
\end{align}
In the case where $\sigma=\tilde{\sigma}$ the outcome $``0"$ must be omitted. 
Finally, we define the corresponding classical distributions by
\begin{align}
    q_0 = \operatorname{Tr}[\sigma-\tilde{\sigma}], \quad 
p_0 = 0, \quad q_{jk} = c_{jk}\operatorname{Tr}[Q_{jk}],
\quad
p_{jk} = r_j\operatorname{Tr}[Q_{jk}].
\end{align}

\end{remark}

Next, we derive an upper bound on the regularized prepared R\'enyi relative entropy in terms of \(D_{\alpha,\alpha-1}\). This bound follows from the above lemma, together with the fact that the additive correction grows only logarithmically with the number of copies and therefore vanishes upon regularization.
\begin{proposition}
\label{prop: achievability}
Let $\rho$ and $\sigma$ be quantum states. Then, for any $\alpha\in [2,\infty)$, it holds that
\begin{align}
    \lim_{n\to\infty}
    \frac1n
    D_\alpha^{\mathbb{P}}
    \left(
        \rho^{\otimes n}\|\sigma^{\otimes n}
    \right)
    \leq 
    D_{\alpha,\alpha-1}(\rho\|\sigma)
\end{align}
\end{proposition}
\begin{proof}
In the case $\operatorname{supp}(\rho)\nsubseteq \operatorname{supp}(\sigma)$ an analogous argument to the one in the proof of Proposition~\ref{prop: converse} shows that both quantities are infinite. Let us consider now the case $\operatorname{supp}(\rho)\subseteq \operatorname{supp}(\sigma)$.
    The upper bound is a consequence of Lemma~\ref{lem: upper bound}. We apply the result to $\rho^{\otimes n}$ and $\sigma^{\otimes n}$. We have
    \begin{align}
        D_\alpha^{\mathbb{P}}
    \left(
        \rho^{\otimes n}\|\sigma^{\otimes n}
    \right)
    &\le
    D_{\alpha,\alpha-1}
    \left(
        \rho^{\otimes n}\|\sigma^{\otimes n}
    \right)
    +
    \log  |\textup{spec}(\rho^{\otimes n})|\leq 
    n D_{\alpha,\alpha-1}
    \left(
        \rho \|\sigma
    \right)
    +
    (d-1)\log (n+1)\,,
    \end{align}
where we used that $D_{\alpha,\alpha-1}$ is additive and that $\log|\textup{spec}(\rho^{\otimes n})| \leq (d-1)\log (n+1)$ and $d$ is the dimension of $\rho$. Multiplying by $1/n$ and taking the limit $n\to \infty$  and noting that the term $((d-1)\log{(n+1)})/n$ goes to zero proves the upper bound.
\end{proof}

The following corollary follows directly from Eq.~\eqref{eq: sandwiched between regularized} and the accompanying discussion in Section~\ref{sec: measured and prepared}.
\begin{corollary}
\label{cor: minimality}
Let $\rho$ and $\sigma$ be quantum states. Then, every quantum relative entropy $\D_\alpha$ that reduces to the classical Rényi relative entropy $D_\alpha$ on classical states satisfies 
\begin{equation}
    \D_\alpha(\rho\|\sigma) \leq  \begin{cases}
    \widehat{D}_{\alpha}(\rho\|\sigma), & \alpha \in [0,2]\\
         D_{\alpha,\alpha-1}(\rho\|\sigma), & \alpha \in (2,\infty]\,.
    \end{cases} 
\end{equation}
\end{corollary}

\section{Multi-copy and catalytic transformations}
As an application, we study transformations between pairs of quantum states. Such pairs, commonly referred to as dichotomies, constitute the basic objects of the resource theory of asymmetric distinguishability~\cite{buscemi2019information,wang2019resource}. We derive transformation criteria in both the many-copy and catalytic settings and subsequently discuss their implications for the transformations of quantum states under Gibbs-preserving maps.
To this end, we specialize the general results for preordered semirings~\cite{fritz2023abstract,fritz2023abstractII} to the present setting, thereby obtaining sufficient conditions for convertibility in both the many-copy and catalytic regimes. These comparison theorems are known as the \emph{Vergleichsstellensätze}.

To apply the general theory of preordered semirings to quantum-state transformations under quantum channels, we first construct the relevant semiring, namely, the quantum-majorization semiring. Our presentation follows~\cite[Section~2.1]{haapasalo2025barycentric}, where a more detailed account can be found.
\subsection{The quantum majorization semiring}
\label{sec: quantum majorization semiring}
A preordered semiring is a tuple $S = (\mathcal{S}, +, \cdot, 0, 1, \rleq)$, where $\mathcal{S}$ is a set equipped with binary operations of addition $+$ and multiplication $\cdot$, a zero element $0 \in \mathcal{S}$, a multiplicative unit $1 \in \mathcal{S}$, and a preorder relation $\rleq$  (a reflexive and transitive binary relation) defined on $\mathcal{S}$  satisfying
\begin{equation}
    x \rleq y \ \Rightarrow \ 
\begin{cases}
x + a \ \rleq\ y + a, \\[0.3em]
x a \ \rleq\ y a,
\end{cases}
\end{equation}
for all $a \in \mathcal{S}$.
Moreover, 
$(\mathcal{S}, +, 0)$ and $(\mathcal{S}, \cdot, 1)$ are commutative semigroups, and the multiplication distributes over the addition.
For compactness, we sometimes omit the multiplication symbol when writing products in the semiring.

We next define the quantum-majorization semiring. The construction is a specialization of that introduced in~\cite{haapasalo2025barycentric}, obtained by restricting from arbitrary tuples of quantum states to pairs.

\begin{enumerate}

\item \textbf{The set of elements ``$\mathcal{S}$"}. 
The elements of the semiring are essentially pairs of states
\((\rho,\sigma)\) with coinciding supports. However, this set must be enlarged
to include non-normalized states, that is, arbitrary positive semidefinite
operators. In addition, pairs that differ only by embedding into a larger space are identified
as equivalent.

We now formalize this construction. Let \(\mathcal{T}_d\) denote the set of positive semidefinite operators on \(\mathbb{C}^d\). For each
\(d \in \mathbb{N}\), define
\begin{align}
    \mathcal{Q}_d
    =
    \left\{
        (\rho,\sigma)\in \mathcal{T}_d \times \mathcal{T}_d
        \,\middle|\,
        \operatorname{supp}(\rho)=\operatorname{supp}(\sigma)
    \right\}.
\end{align}
We then set
\begin{align}
    \mathcal{Q}
    =
    \bigcup_{d=1}^{\infty} \mathcal{Q}_d,
\end{align}
so that \(\mathcal{Q}\) collects all such pairs of positive semidefinite
operators over arbitrary finite-dimensional Hilbert spaces.

Let \((\rho,\sigma)\in\mathcal{Q}_d\) and
\((\rho',\sigma')\in\mathcal{Q}_{d'}\). We write $(\rho,\sigma)\approx(\rho',\sigma')$
if there exist an integer \(\ell\geq \max\{d,d'\}\) and isometries $U\colon\mathbb{C}^d\to\mathbb{C}^{\ell}$, and $V\colon\mathbb{C}^{d'}\longrightarrow\mathbb{C}^{\ell}$
such that $U\rho U^\dagger = V\rho' V^\dagger$ and $ U\sigma U^\dagger = V\sigma' V^\dagger$.
We denote by \([(\rho,\sigma)]\) the equivalence class of \((\rho,\sigma)\) under \(\approx\). The semiring set $\mathcal{S}$ is then the set of equivalence classes of pairs of states under this relation, namely 
\begin{align}
    \mathcal{S}= \mathcal{Q}/{\approx}\,.
\end{align}

\item \textbf{The addition ``$+$"}.
Addition is defined componentwise through the direct sum of operators. Given
\((\rho,\sigma)\in\mathcal{Q}_d\) and
\((\tau,\eta)\in\mathcal{Q}_{d'}\), we set
\begin{align}
    (\rho,\sigma)\boxplus(\tau,\eta)
    =
    (\rho\oplus\tau,\sigma\oplus\eta).
\end{align}
This operation induces an addition on the quotient
\(\mathcal{Q}/{\approx}\) by
\begin{align}
    [(\rho,\sigma)]+[(\tau,\eta)]
    =
    [(\rho\oplus\tau,\sigma\oplus\eta)].
\end{align}
In other words, the sum of two equivalence classes is the equivalence class
of the componentwise direct sum of any pair of representatives.

\item \textbf{The multiplication ``\(\cdot\)''.}
Multiplication is defined componentwise through the tensor product of
operators. Given
\((\rho,\sigma)\in\mathcal{Q}_d\) and
\((\tau,\eta)\in\mathcal{Q}_{d'}\), we set
\begin{align}
    (\rho,\sigma)\boxtimes(\tau,\eta)
    =
    (\rho\otimes\tau,\sigma\otimes\eta).
\end{align}
This operation induces a multiplication on the quotient
\(\mathcal{Q}/{\approx}\) according to
\begin{align}
    [(\rho,\sigma)]\cdot[(\tau,\eta)]
    =
    [(\rho\otimes\tau,\sigma\otimes\eta)].
\end{align}
In other words, the product of two equivalence classes is the equivalence
class of the componentwise tensor product of any pair of representatives.

\item \textbf{The preorder ``\(\rgeq\)''.}
The preorder is defined by convertibility of pairs under quantum
channels. Given
\((\rho,\sigma)\in\mathcal{Q}_d\) and
\((\rho',\sigma' )\in\mathcal{Q}_{d'}\), we write $(\rho,\sigma)\succeq(\rho',\sigma')$
if there exists a completely positive trace-preserving map $\mathcal{E}:\mathcal{T}_d\to \mathcal{T}_{d'} $
such that $\mathcal{E}(\rho)=\rho'$ and $ \mathcal{E}(\sigma)=\sigma'$.

It induces a preorder on the quotient \(\mathcal{Q}/{\approx}\) by
\begin{align}
    [(\rho,\sigma)]\rgeq[(\rho',\sigma')]
    \quad\Longleftrightarrow\quad
    (\rho,\sigma)\succeq(\rho',\sigma').
\end{align}

\item \textbf{The zero element ``\(0\)''.}
The zero element is defined as
\begin{align}
    0 = [(0,0)],
\end{align}
where \((0,0)\) is the pair of zero operators on the one-dimensional
Hilbert space.

\item \textbf{The unit element ``\(1\)''.}
The unit element is defined as
\begin{align}
    1 = [(1,1)],
\end{align}
where \((1,1)\) is the pair of identity operators on the one-dimensional
Hilbert space.
\end{enumerate}

The  \emph{quantum majorization semiring} is therefore
\begin{align}
    \bigl(
        \mathcal{Q}/{\approx},+,
        \cdot,
        [(0,0)],
        [(1,1)],
        \lgeq
    \bigr).
\end{align}

We conclude this subsection by introducing a few additional concepts needed to apply the results on many-copy and catalytic transformations.

We say that a preordered semiring $S$ is a preordered semidomain if
\begin{equation}
    \begin{aligned}
&xy = 0 \ \Rightarrow\ x = 0 \ \text{or} \ y = 0, \\
&0 \rleq x \rleq 0 \ \Rightarrow\ x = 0.
\end{aligned}
\end{equation}
In addition, $S$ is zerosumfree if $x + y = 0$ implies $x = 0 = y$. 
It is easy to verify that the quantum majorization semiring is a zerosumfree preordered semidomain.

\subsubsection{Power universals}
A preordered semiring $S$ has polynomial growth if it admits a power universal element $u \in \mathcal{S}$, that is,
\begin{equation}
    x \rleq y \quad \Rightarrow \quad \exists\, k \in \mathbb{N}:\ y \rleq x u^k.
\end{equation}
The existence of a power universal element is a key requirement for applying
Theorem~\ref{thm:Fritz2022}, stated below, to many-copy and catalytic
transformations of pairs of states. The following lemma provides an explicit
characterization of such elements. In particular, it gives a concrete
condition on the input pair under which the sufficiency direction of
Theorem~\ref{thm:Fritz2022} can be invoked to establish the existence of a
quantum channel realizing the desired transformation in the many-copy
regime. The following lemma, established in~\cite[Lemma 16]{haapasalo2025barycentric}, characterizes the power-universal elements of the quantum-majorization semiring.
\begin{lemma}
\label{lem:power-universal}
The quantum majorization semiring \(S\) is of polynomial growth. Moreover,
an element \([(\rho,\sigma)]\in S\) is power universal if and only if $ \Tr[\rho]=\Tr[\sigma]=1$ and $\rho\neq \sigma$.
\end{lemma}

\subsubsection{A surjective homomorphism}

To apply Theorem~\ref{thm:Fritz2022}, stated below, we first identify a
surjective homomorphism with trivial kernel,
\begin{align}
    \lVert\,\cdot\,\rVert \colon S
    \longrightarrow
    \mathbb{R}_{>0}^{d} \cup \{(0,\ldots,0)\},
\end{align}
satisfying
\begin{align}
    &a \rgeq b
    \quad\Longrightarrow\quad
    \lVert a\rVert = \lVert b\rVert,
    \\
    &\lVert a\rVert = \lVert b\rVert
    \quad\Longrightarrow\quad
    a \sim b.
\end{align}
Here, \(x\sim y\) denotes the equivalence relation induced by the preorder
\(\rleq\). Explicitly, \(x\sim y\) if and only if there exist
\(z_1,\ldots,z_n\in S\) such that
\begin{align}
    x \rleq z_1 \rgeq z_2 \rleq \cdots \rgeq z_n \rleq y.
\end{align}
We denote the component homomorphisms of
\(\lVert\,\cdot\,\rVert\) by
\(\lVert\,\cdot\,\rVert_{(j)}\), for \(j=1,\ldots,d\).

For the quantum-majorization semiring of pairs of possibly unnormalized
quantum states, the required homomorphism is obtained by taking the trace of
each element of the pair. Explicitly,
\begin{align}
\bigl\lVert[(\rho,\sigma)]\bigr\rVert
    &=
    \bigl(
        \operatorname{Tr}[\rho],
        \operatorname{Tr}[\sigma]
    \bigr).
\end{align}
The kernel is trivial, since
\begin{align}
    \bigl\lVert[(\rho,\sigma)]\bigr\rVert=(0,0)
    &\quad\Longleftrightarrow\quad
    \operatorname{Tr}[\rho]=\operatorname{Tr}[\sigma]=0
    \\
    &\quad\Longleftrightarrow\quad
    \rho=\sigma=0.
\end{align}
Furthermore, if
\([(\rho,\sigma)]\rgeq[(\rho',\sigma')]\), then
\begin{align}
    \bigl\lVert[(\rho,\sigma)]\bigr\rVert
    &=
    \bigl\lVert[(\rho',\sigma')]\bigr\rVert,
\end{align}
because \(\rho'\) and \(\sigma'\) are obtained from \(\rho\) and \(\sigma\),
respectively, through a trace-preserving map.

It remains to verify the second condition. Suppose that
\begin{align}
    \bigl\lVert[(\rho,\sigma)]\bigr\rVert
    &=
    \bigl\lVert[(\rho',\sigma')]\bigr\rVert.
\end{align}
Equivalently,
\begin{align}
    \operatorname{Tr}[\rho]
    &=
    \operatorname{Tr}[\rho'],\quad 
    \operatorname{Tr}[\sigma]
    =
    \operatorname{Tr}[\sigma'].
\end{align}
Since the trace is a completely positive and trace-preserving map, we obtain 
\begin{align}
    [(\rho,\sigma)]
    &\rgeq
    \bigl[
        \bigl(\operatorname{Tr}[\rho],
              \operatorname{Tr}[\sigma]\bigr)
    \bigr]
    =
    \bigl[
        \bigl(\operatorname{Tr}[\rho'],
              \operatorname{Tr}[\sigma']\bigr)
    \bigr]
    \lgeq
    [(\rho',\sigma')].
\end{align}
Consequently,
\begin{align}
    \bigl\lVert[(\rho,\sigma)]\bigr\rVert
    =
    \bigl\lVert[(\rho',\sigma')]\bigr\rVert
    \quad\Longrightarrow\quad
    [(\rho,\sigma)]
    \sim
    [(\rho',\sigma')].
\end{align}
Thus, \(\lVert\,\cdot\,\rVert\) satisfies all the required properties. In the
present setting, it has two component homomorphisms, given by
\begin{align}
    \bigl\lVert[(\rho,\sigma)]\bigr\rVert_{(1)}
    =
    \operatorname{Tr}[\rho],
    \quad 
    \bigl\lVert[(\rho,\sigma)]\bigr\rVert_{(2)}
    =
    \operatorname{Tr}[\sigma].
\end{align}

\subsubsection{Monotone homomorphisms and derivations}
\label{subsec: homo and deri}
Within the framework of preordered semirings, transformation criteria for products of elements are formulated in terms of monotone homomorphisms and monotone derivations. In the specific case of the quantum majorization semiring, as we discuss below, these functionals correspond,  after taking logarithms and applying suitable normalizations, to R\'enyi relative entropies and govern transformations in both the many-copy and uncorrelated catalytic regimes. In the following, we first introduce monotone homomorphisms and derivations formally and then explain their relationship with the R\'enyi relative entropies.

Given preordered semirings $S$ and $T$, we say that a map $\Phi:\mathcal{S}\to \mathcal{T}$ is a {\it monotone homomorphism} if it satisfies the properties
\begin{enumerate}
\item Additivity: $\Phi(x+y)=\Phi(x)+\Phi(y)$ for all $x,y\in \mathcal{S}$,
\item Multiplicativity: $\Phi(xy)=\Phi(x)\Phi(y)$ for all $x,y\in S$,
\item Monotonicity: $x\rleq y$ $\Rightarrow$ $\Phi(x)\rleq\Phi(y)$, and
\item $\Phi(0)=0$ and $\Phi(1)=1$.
\end{enumerate}
We call a monotone homomorphism \textit{degenerate} if $ x \rleq y \quad \Rightarrow \quad \Phi(x) = \Phi(y)$. Otherwise, it is called \textit{nondegenerate}.
We need to consider monotone homomorphisms with values in certain special semirings, namely:
\begin{enumerate}
    \item $\mathbb{R}_+$: the half-line $[0,+\infty)$ with the usual addition, multiplication, and total order.
    \item $\mathbb{R}_+^{\mathrm{op}}$: the same set with reversed order. 
    \item $\mathbb{T}\mathbb{R}_+$: the half-line $[0,+\infty)$ with the usual multiplication, order, and the tropical sum $x + y = \max\{x,y\}$.
    \item $\mathbb{T}\mathbb{R}_+^{\mathrm{op}}$: the same as for $\mathbb{T}\mathbb{R}_+$, but with reversed order.
\end{enumerate}
The term \emph{temperate reals} refers to the pair $(\mathbb{R}_+, \mathbb{R}_+^{\mathrm{op}})$, while \emph{tropical reals} denotes the pair $(\mathbb{T}\mathbb{R}_+, \mathbb{T}\mathbb{R}_+^{\mathrm{op}})$.
Given a monotone homomorphism $\Phi: \mathcal{S} \to \mathbb{R}_+$, an additive map $\Delta: \mathcal{S} \to \mathbb{R}$ is called a derivation at $\Phi$ (or a $\Phi$-derivation) if it satisfies the Leibniz rule
\begin{equation}
    \Delta(xy) = \Delta(x)\,\Phi(y) + \Phi(x)\,\Delta(y)
\end{equation}
for all $x, y \in \mathcal{S}$.
In this work, we are interested in $\Phi$-derivations at degenerate homomorphisms that are also monotone, i.e., such that
\begin{equation}
    x \rleq y \quad \Rightarrow \quad \Delta(x) \leq \Delta(y).
\end{equation}

\bigskip

A complete characterization of the monotone homomorphisms and derivations of the quantum-majorization semiring is not currently known. Such a characterization is not required for our purposes, however, because we restrict our attention to transformations whose input pair is classical. As shown below, in this setting it suffices to consider the maximal quantum extensions of the corresponding classical functionals.
By Theorem~\ref{thm: main theorem}, these maximal extensions are the geometric R\'enyi relative entropies introduced in~\eqref{def: geometric relative entropy} for \(\alpha\in[0,2]\), and the \(\alpha\)-\(z\) R\'enyi relative entropies with \(z=\alpha-1\), introduced in~\eqref{def: alpha-z for z=alpha-1}, for \(\alpha>2\).
We note that, since these quantities are invariant under isometric embeddings, their values depend only on the equivalence class of the underlying pair. They therefore induce well-defined functionals on the quotient space by assigning to each equivalence class the value attained on any of its representatives.

From the defining properties above, it follows that monotone homomorphisms and derivations of the quantum-majorization semiring can be identified, up to taking logarithms and applying constant normalization factors, with relative entropies. Whenever no confusion can arise, we therefore identify each homomorphism or derivation with its corresponding relative entropy. For every relative entropy listed below, the functional obtained by exchanging its two arguments is also implicitly included. We now present the homomorphisms and derivations that are relevant for our purposes, grouped according to their codomains.
\begin{itemize}
    \item \(\mathbb{R}_{+}^{\mathrm{op}}\):
    the geometric R\'enyi relative entropy
    \(\widehat{D}_{\alpha}\) for
    \(\alpha\in(0,1)\).

    \item \(\mathbb{R}_{+}\):
    the geometric R\'enyi relative entropy
    \(\widehat{D}_{\alpha}\) for
    \(\alpha\in(1,2]\), and the
    \(\alpha\)-\(z\) R\'enyi relative entropy
    \(D_{\alpha,\alpha-1}\) for
    \(\alpha>2\).

    \item \(\mathbb{TR}_{+}\):
    the max-relative entropy \(D_{\max}\), obtained as the limit $\alpha\to\infty$ of $D_{\alpha,\alpha-1}$.
    \item \(\mathbb{TR}_{+}^{\mathrm{op}}\):
    there are no nondegenerate monotone homomorphisms into this
    target semiring; see~\cite[Lemma~17]{haapasalo2025barycentric}.

    \item Derivations: The Belavkin--Staszewski relative entropy $\widehat{D}$, obtained as the limit \(\alpha \to 1\) of the geometric Rényi relative entropy. When evaluated on the pair \((\rho,\sigma)\), it gives the derivation associated with the first component of the surjective homomorphism $\lVert\,\cdot\,\rVert_{(1)}$; exchanging the two arguments gives the derivation associated with the second component $\lVert\,\cdot\,\rVert_{(2)}$.
\end{itemize}
Finally, the above list is exhaustive in the classical setting, namely, for pairs of probability distributions with coinciding supports.
Indeed, the R\'enyi relative entropies for $\alpha\in (0,\infty]$ exhaust all monotone homomorphisms and all (extremal) monotone derivations of the corresponding classical semiring; see~\cite[Propositions~13 and~14]{farooq2024matrix}.

\subsection{Main theorem on multiple copies and catalytic transformation of states}
In this section, we derive our main results on the transformation of multiple copies of pairs of quantum states and on uncorrelated catalytic transformations. The key result that we will use is the following.

\begin{theorem}[Based on Theorem 8.6 in \cite{fritz2023abstractII}]\label{thm:Fritz2022}
Let $S$ be a zerosumfree preordered semidomain with a power universal element $u$. Assume that for some $d\in\mathbb N$ there is a surjective homomorphism $\|\cdot\|:\mathcal{S}\to \mathbb R_{>0}^d\cup\{(0,\ldots,0)\}$ with trivial kernel and such that
\begin{equation}\label{eq:surjectivehomomorphismproperties}
a\rgeq b\ \Rightarrow\ \|a\|=\|b\|\quad {\rm and} \quad \|a\|=\|b\|\ \Rightarrow\ a\sim b.
\end{equation}
Denote the component homomorphisms of $\|\cdot\|$ by $\|\cdot\|_{(j)}$, $j=1,\ldots,d$. Let $x,y\in S\setminus\{0\}$ with $\|x\|=\|y\|$. If 
\begin{itemize}
\item[(i)] for every $\mathbb K\in\{\mathbb R_+,\mathbb R_+^{\rm op},\mathbb T \mathbb R_+,\mathbb T\mathbb R_+^{\rm op}\}$ and every nondegenerate monotone homomorphism $\Phi : \mathcal{S} \to \mathbb K$ with trivial kernel, we have $\Phi(x) > \Phi(y)$, and
\item[(ii)] $\Delta(x) > \Delta(y)$ for every monotone $\|\cdot\|_{(j)}$-derivation $\Delta : \mathcal{S} \to \mathbb R$ with $\Delta(u) = 1$ for all component indices $j = 1,\ldots,d$,
\end{itemize}
then
\begin{enumerate}[label=(\alph*), ref=2(\alph*)]
    \item there exists a nonzero $c\in \mathcal{S}$ such that $cx\rgeq cy$, and 
\label{it: catalytic}
\item if additionally $x$ is power universal, then $x^n \rgeq y^n$ for all sufficiently large $n\in\mathbb N$. 
\label{it: asymptotic}
\end{enumerate}
Conversely, if either of these properties holds for some integer $n \geq 1$ or a catalyst $c$, then the inequalities in items (i) and (ii) above hold non-strictly. 
\end{theorem}
Large-sample ordering as described in item \textit{(b)} of Theorem~\ref{thm:Fritz2022} implies catalytic ordering as in item \textit{(a)} of Theorem~\ref{thm:Fritz2022}, where the catalyst can be chosen as
\begin{equation}
    c = \sum_{\ell=0}^{n-1} x^\ell y^{n-1-\ell}
\end{equation}
for sufficiently large $n \in \mathbb{N}$. This implication was originally established in \cite{duan2005multiple} for the case where $x, y$ are probability vectors, but the argument extends naturally to the more general framework considered here.
The converse, however, generally fails to hold. In particular,~\cite[Theorem 3]{feng2006relation} presents a method for constructing explicit counterexamples within a specific context.

\bigskip

We now specialize the above theorem to the current quantum majorization semiring introduced above. 
Since in the following we consider classical input pairs, many of the general conditions simplify. Indeed, it turns out that it is enough to verify that the classical Rényi relative entropies of the input pair are larger than the corresponding maximal quantum extensions evaluated on the output pair. A similar discussion has also appeared in~\cite[Section~6]{verhagen2025matrix}.

\begin{theorem}
\label{thm: large sample and catalytic}
   Let $(p, q)$ be a pair of classical states and $(\rho, \sigma)$ be a pair of quantum states, each having coinciding support. If
    \begin{align}
    & D_\alpha(p\| q)> \widehat{D}_\alpha(\rho\| \sigma), \qquad\qquad\qquad   \forall \alpha\in [1/2,2], \\
    &D_\alpha(q\| p)> \widehat{D}_\alpha(\sigma\| \rho), \qquad\qquad\qquad   \forall \alpha\in [1/2,2], \\
    &D_\alpha(p\| q)> D_{\alpha,\alpha-1}(\rho\| \sigma), \qquad\;\qquad   \forall \alpha \in [2,\infty], \\
    &D_\alpha(q\| p)> D_{\alpha,\alpha-1}(\sigma\| \rho), \qquad\qquad\;   \forall \alpha \in [2,\infty]\,.
    \end{align}
    then
   \begin{enumerate}
    \item for sufficiently large \(n \in \mathbb{N}\), there exists a quantum channel $\mathcal{E}_n$ such that  
    \begin{align}
        \mathcal{E}_n(p^{\otimes n}) = \rho^{\otimes n} \,, \text{and} \quad \mathcal{E}_n(q^{\otimes n}) = \sigma^{\otimes n}\,.
    \end{align}
    \item there exist quantum states \(\nu,\eta\) with coinciding support and a channel $\mathcal{F}$ such that
    \begin{equation}
       \mathcal{F}(p \otimes \nu) = \rho \otimes \nu \,, \text{and} \quad \mathcal{F}(q \otimes \eta) = \sigma \otimes \eta\,.
    \end{equation}
\end{enumerate}
Conversely, the existence of such a channel $\mathcal{E}_n$ (for some $n \geq 1$) or $\mathcal{F}$ implies that the aforementioned inequalities hold non-strictly.
\end{theorem}

\begin{proof}
By Theorem~\ref{thm:Fritz2022}, the strict-ordering conditions must be verified
for every monotone homomorphism and every derivation of the
quantum-majorization semiring. However, when restricted to commuting pairs, these functionals depend only on the
corresponding classical probability distributions and do not depend on the choice of
the joint eigenbasis. Indeed, two representations of the same classical pair are
related by a simultaneous unitary conjugation, possibly after an isometric
embedding. Monotonicity under the corresponding unitary or isometric channels,
together with monotonicity under their inverses on the relevant supports, then
implies that the two representations have the same value.

Moreover, on commuting pairs, the semiring operations and the channel-induced
preorder reduce to their classical counterparts: direct sums and tensor
products become the corresponding operations on probability vectors, while
quantum channels reduce to stochastic maps. Consequently, the restrictions of
the quantum homomorphisms and derivations coincide with those of the
matrix-majorization semiring studied in~\cite{farooq2024matrix}; see
also~\cite[Theorem~1]{mu2021blackwell}.

For pairs of probability distributions with coinciding supports, these
classical functionals have been completely characterized. Up to taking a
logarithm and multiplying by a positive constant, they are precisely
\begin{align}
    D_\alpha(p\|q)
    \qquad\text{and}\qquad
    D_\alpha(q\|p),
    \qquad
    \alpha\in[1/2,\infty];
\end{align}
see~\cite[Propositions~13 and~14 and Corollary~23]{farooq2024matrix}.
The case \(\alpha=1\) corresponds to the two derivations, whereas
\(\alpha=\infty\) corresponds to homomorphisms with values in the tropical
reals. The remaining parameter values correspond to homomorphisms with values
in \(\mathbb{R}_{+}\) or \(\mathbb{R}_{+}^{\mathrm{op}}\); see also
Section~\ref{subsec: homo and deri}.

Accordingly, every quantum extension \(\mathbb{D}_{\alpha}\) of the classical Rényi relative entropy \(D_{\alpha}\), with \(\alpha \in [1/2,\infty]\), arising from a monotone homomorphism or derivation must satisfy
\begin{align}
   \mathbb  D_\alpha(p\|q)
    >
    \mathbb{D}_\alpha(\rho\|\sigma)\,,\qquad \mathbb D_\alpha(q\|p)
    &>
    \mathbb{D}_\alpha(\sigma\|\rho)\,.
\end{align}

For a fixed value of \(\alpha\), several inequivalent quantum extensions
may have the same classical restriction. Let \(\mathbb{D}_\alpha\) be
any additive quantum relative entropy whose restriction to commuting
states is \(D_\alpha\). Since the input pair \((p,q)\) is classical, its
value is independent of the chosen extension $\mathbb{D}_\alpha(p\|q)
    =
    D_\alpha(p\|q)$.
A minor support issue must be addressed when applying the
data-processing inequality. The test spectrum underlying
Theorem~\ref{thm:Fritz2022} is defined on the semiring of classical
pairs with coinciding supports or, equivalently, with full support
after common-zero outcomes have been removed. Hence, for a preparation
channel \(\mathcal{F}\) satisfying $\mathcal{F}(p)
    =
    \rho$, $\mathcal{F}(q)
    =
    \sigma$ and $\operatorname{supp}(p)
    =
    \operatorname{supp}(q)$
the data-processing inequality gives
\begin{align}
    \mathbb{D}_\alpha(\rho\|\sigma)
    &=
    \mathbb{D}_\alpha\bigl(
        \mathcal{F}(p)
        \big\|
        \mathcal{F}(q)
    \bigr)
    \leq
    \mathbb{D}_\alpha(p\|q)
    =
    D_\alpha(p\|q).
    \label{eq:spectral-extension-preparation-bound}
\end{align}
Taking the infimum over all such preparations yields
\begin{align}
    \mathbb{D}_\alpha(\rho\|\sigma)
    &\leq
    \inf_{\substack{
        p,q,\mathcal{F}\\
        \mathcal{F}(p)=\rho,\,
        \mathcal{F}(q)=\sigma\\
        \operatorname{supp}(p)=\operatorname{supp}(q)
    }}
    D_\alpha(p\|q).
    \label{eq:spectral-extension-restricted-prepared}
\end{align}
This does not immediately give an upper bound in terms of the prepared relative entropy, since the latter is defined without the condition $\operatorname{supp}(p)=\operatorname{supp}(q)$
and therefore it generally gives only a smaller value.
However, Lemma~\ref{lem:prepared-coinciding-support} shows that the additional condition on the supports does not change the value of the infimum.
Dividing by \(n\) and regularizing therefore gives by Theorem~\ref{thm: main theorem}
\begin{align}
    \D_\alpha(\rho\Vert \sigma)
\leq 
    \begin{cases}
        \widehat{D}_\alpha(\rho\|\sigma),
        & \alpha\in[1/2,2],
        \\[1mm]
        D_{\alpha,\alpha-1}(\rho\|\sigma),
        & \alpha\in(2,\infty].
    \end{cases}
    \label{eq:maximal-extension-closed-form}
\end{align}
The same argument applies with the two arguments exchanged. Therefore,
because the classical input value \(D_\alpha(p\|q)\) is independent of
the chosen extension, it is sufficient to impose the inequalities for the maximal ones. By maximality, they then hold automatically for
every other admissible extension. The abstract conditions of
Theorem~\ref{thm:Fritz2022} consequently reduce to inequalities
involving the maximal relative entropies identified above.

In addition, according to~\cite[Theorem~8.6]{fritz2023abstractII}, we need to verify that the input pair $(p,q)$ is
power universal. This fact is already guaranteed by the above assumptions on the strict
ordering of entropies.
Indeed, all relative entropies are non-negative. Since the classical Rényi relative entropies are faithful for $\alpha>0$, the strict inequalities assumed in the theorem imply that $p\neq q$ and hence, by Lemma~\ref{lem:power-universal}, $(p,q)$ must be power universal.

Finally, all these quantities satisfy the data-processing inequality throughout their respective parameter ranges~\cite{matsumoto2018maxdivergence,zhang2020wigner}. Consequently, the existence of a channel \(\mathcal{E}_n\), for some \(n \geq 1\), or of a channel \(\mathcal{F}\) necessarily implies that the inequalities above hold in their non-strict form.
\end{proof}

\bigskip
\bigskip

The above result implies that the maximal relative entropies determine the optimal transformation rate between pairs of states, thereby providing them with an operational interpretation as characterizing the number of output copies obtainable per input copy.

Specifically, given two pairs $(p,q)$ and $(\rho,\sigma)$, we seek the largest $r \in \mathbb{R}$ for which there exists a channel $ \mathcal{E}_n$ such that $\mathcal{E}_n (p^{\otimes n})=\rho^{\otimes \lfloor r n \rfloor}$ and $\mathcal{E}_n (q^{\otimes n})=\sigma^{\otimes \lfloor r n \rfloor}$ hold for all sufficiently large $n$, i.e., we want to determine the value
\begin{equation}
    R((p,q)\rightarrow (\rho,\sigma)) =\; \sup\left\{ r \ge 0 \;\middle|\; \mathcal{E}_n (p^{\otimes n})=\rho^{\otimes \lfloor r n \rfloor}, \mathcal{E}_n (q^{\otimes n})=\sigma^{\otimes \lfloor r n \rfloor} \;\; \text{for large } n \right\}.
\end{equation}
It is known that this problem is equivalent to the large sample problem of transforming multiple copies (see e.g.~\cite[Theorem 3.11]{Jensen_Kjaerulf_2019},~\cite[Section 5]{verhagen2025matrix}). In particular, we use the necessary and sufficient conditions for a large sample in Theorem~\ref{thm: large sample and catalytic} to obtain
\begin{corollary}
\label{cor: rate}
Let $(p,q)$ be a pair of classical states and $(\rho,\sigma)$ a pair
of quantum states, each with coinciding supports. Then,
\begin{align}
R\bigl((p,q)\to(\rho,\sigma)\bigr)
=
\min\Biggl\{&
\min_{\alpha\in[1/2,2]}
\left\{
\frac{D_\alpha(p\|q)}
     {\widehat D_\alpha(\rho\|\sigma)},
\frac{D_\alpha(q\|p)}
     {\widehat D_\alpha(\sigma\|\rho)}
\right\},
\min_{\alpha\in[2,\infty]}
\left\{
\frac{D_\alpha(p\|q)}
     {D_{\alpha,\alpha-1}(\rho\|\sigma)},
\frac{D_\alpha(q\|p)}
     {D_{\alpha,\alpha-1}(\sigma\|\rho)}
\right\}
\Biggr\}.
\end{align}
\end{corollary}

\subsection{Catalytic generation of coherence under Gibbs-preserving maps}
\label{sec:quantum-thermodynamics}
In this section, we establish necessary and sufficient conditions for the
catalytic conversion of energy-incoherent states into energy-coherent
states under Gibbs-preserving maps. Such maps are known to generate
coherence between distinct energy levels from initially energy-incoherent
states~\cite{faist2015gibbs}. The conditions derived here completely
characterize this capability in the catalytic setting.

Given a Hamiltonian $H$, the Gibbs state is defined as
\begin{align}
    \gamma=\frac{1}{Z}e^{-\beta H},
\end{align}
where $\beta$ is the inverse temperature and
$Z=\Tr[e^{-\beta H}]$ is the partition function. In the following, we always assume that $\beta>0$.
A state is called \emph{energy-incoherent} if it is block diagonal with
respect to the energy eigenspaces, or equivalently, if it commutes with
the Hamiltonian.
Since every energy-incoherent state commutes with $H$, the two operators
can be simultaneously diagonalized. Thus, we may choose a common
orthonormal eigenbasis $\{\lvert E_i\rangle\}_{i=1}^d$ and write
$H=\sum_{i=1}^d E_i\lvert E_i\rangle\langle E_i\rvert$ and the state as
$\sum_{i=1}^d p_i\lvert E_i\rangle\langle E_i\rvert$, for some
probability distribution $p=\{p_i\}_{i=1}^d$. Throughout this section,
we identify $p$ with the corresponding diagonal state and use the same
symbol for both. Likewise, we occasionally identify the Gibbs state
$\gamma$ with the probability distribution
$\gamma=\{\gamma_i\}_{i=1}^d$. With this convention,
$D_\alpha(p\|\gamma)$ and $D_\alpha(\gamma\|p)$ denote the classical
R\'enyi divergences between the corresponding probability distributions.
A state is called \emph{energy-coherent} if it is not diagonal in any
energy eigenbasis, or equivalently, if it does not commute with $H$.

A quantum channel $\mathcal{F}$ is called Gibbs-preserving if it leaves the Gibbs state invariant, i.e., $\mathcal{F}(\gamma) = \gamma$. Gibbs-preserving operations can generate coherence; that is, they can map an energy-incoherent state to an energy-coherent state~\cite{faist2015gibbs}. Consequently, they constitute a strictly larger set than the physically motivated thermal operations, which are incapable of generating coherence among energy levels~\cite{janzing2000thermodynamic,horodecki2013fundamental,brandao2013resource}.

We begin by formally defining state transformations assisted by an uncorrelated catalyst.
\begin{definition}
Let $\rho$ and $\rho'$ be quantum states. We say that $\rho$ can be transformed into $\rho'$ by a catalytic operation if, for every $\varepsilon>0$, there exist a Gibbs-preserving operation $\mathcal{F}_\varepsilon$ and a catalyst state $\nu_\varepsilon$ and a quantum state $\rho_\varepsilon'$ such that $\mathcal{F}_\varepsilon\bigl(\rho\otimes\nu_{\varepsilon}\bigr) = \rho'_\varepsilon\otimes\nu_{\varepsilon}$ and $\frac{1}{2}\bigl\|\rho'_\varepsilon-\rho'\bigr\|_1 \leq \varepsilon$.
\end{definition}

The following result arises as a direct consequence of Theorem~\ref{thm: large sample and catalytic} by setting $q=\sigma=\gamma$.

\begin{theorem}
\label{thm:catalytic-gibbs-preserving}
Let $p$ be a full-rank state that is diagonal in the energy eigenbasis, and let $\rho$ be a full-rank quantum state. Then, $p$ can be transformed into $\rho$ by a catalytic operation if and only if
\begin{align}
    D_\alpha(p\|\gamma) &\geq \widehat{D}_\alpha(\rho\|\gamma), &&\forall\,\alpha\in[1/2,2], \label{eq:gpo-forward-geometric}\\
    D_\alpha(\gamma\|p) &\geq \widehat{D}_\alpha(\gamma\|\rho), &&\forall\,\alpha\in[1/2,2], \label{eq:gpo-reverse-geometric}\\
    D_\alpha(p\|\gamma) &\geq D_{\alpha,\alpha-1}(\rho\|\gamma), &&\forall\,\alpha\in[2,\infty], \label{eq:gpo-forward-alpha-z}\\
    D_\alpha(\gamma\|p) &\geq D_{\alpha,\alpha-1}(\gamma\|\rho), &&\forall\,\alpha\in[2,\infty]. \label{eq:gpo-reverse-alpha-z}
\end{align}
\end{theorem}

\begin{proof}
We first prove sufficiency. If $\rho=\gamma$, the claim follows trivially via the Gibbs-replacement channel with a trivial catalyst. We may therefore assume that $\rho\neq\gamma$.
For any $\varepsilon>0$, we define
\begin{equation}
    \rho_{\varepsilon} = (1-\varepsilon)\rho+\varepsilon\gamma.
    \label{eq:gibbs-perturbed-target}
\end{equation}
By construction, this state satisfies $\frac{1}{2}\bigl\|\rho_\varepsilon-\rho\bigr\|_1 \leq\varepsilon$. We next demonstrate that replacing $\rho$ with $\rho_{\varepsilon}$ renders all inequalities in \eqref{eq:gpo-forward-geometric}--\eqref{eq:gpo-reverse-alpha-z} strict. Let $\mathcal{D}$ denote the collection of relative entropies $\widehat{D}_{\alpha}$ for $\alpha\in[1/2,2]$ and $D_{\alpha,\alpha-1}$ for $\alpha\in[2,\infty]$, with both possible orderings of the arguments included.
For any $\mathbb{D}_\alpha \in \mathcal{D}$, we obtain
\begin{align}
    \mathbb{D}_\alpha(\rho \| \gamma) > \mathbb{D}_\alpha(\rho_\varepsilon \| \gamma).
\end{align}
This strict inequality follows from the joint concavity of the map $\mathbb{Q}_\alpha(\rho\|\sigma) = \exp\!\bigl((\alpha-1)\mathbb{D}_\alpha(\rho\|\sigma)\bigr)$ for $\alpha\in (0,1)$ and its joint convexity for $\alpha \in (1,\infty)$, combined with the assumption that $\rho\neq \gamma$ and the faithfulness of the relative entropy for all $\alpha>0$.  At $\alpha=1$, the conclusion follows
directly from the joint convexity and faithfulness of the
Belavkin--Staszewski relative entropy. At $\alpha=\infty$ one considers the trace term in the definition of the max-relative entropy.

To illustrate this explicitly, consider the case $\alpha>1$. Then, joint convexity gives
\begin{align}
    \mathbb{Q}_\alpha(\rho_\varepsilon\|\gamma)
    &\leq
    \left(1-\varepsilon \right)
    \mathbb{Q}_\alpha\!\left(\rho \middle\| \gamma \right)
    + \varepsilon <
    \mathbb{Q}_\alpha\!\left(\rho \middle\| \gamma \right),
\end{align}
where we used that the assumption $\rho\neq \gamma$ and the faithfulness of the relative entropy imply $ Q_\alpha\!\left(\rho \middle\| \gamma \right)
>1$.
By taking the logarithm, we obtain the desired strict inequality.

The assumed non-strict inequalities \eqref{eq:gpo-forward-geometric}--\eqref{eq:gpo-reverse-alpha-z}, together with the preceding argument, ensure that the conditions of Theorem~\ref{thm: large sample and catalytic} are satisfied for the input pair $(p,\gamma)$ and the output pair $(\rho_{\varepsilon},\gamma)$. Consequently, for any $\varepsilon>0$, there exist full-rank states $\nu_{\varepsilon}$ and $\eta_\varepsilon$, along with a channel $\mathcal{F}_\varepsilon$, such that
\begin{equation}
    \mathcal{F}_\varepsilon(p\otimes\nu_\varepsilon) = \rho_{\varepsilon}\otimes\nu_{\varepsilon}, \quad \text{and} \quad \mathcal{F}_\varepsilon(\gamma\otimes\eta_\varepsilon) = \gamma\otimes\eta_\varepsilon.
    \label{eq:gpo-catalytic-gibbs-exact}
\end{equation}
By choosing the catalyst's Hamiltonian such that its corresponding Gibbs state is $\eta_\varepsilon$, we obtain the desired Gibbs-preserving map.

Conversely, suppose that the approximate conversion is possible, and fix
$\mathbb{D}_{\alpha}\in\mathcal{D}$. Since $\eta_{\varepsilon}$ is a
full-rank Gibbs state,
$\mathbb{D}_{\alpha}(\nu_{\varepsilon}\|\eta_{\varepsilon})$ is finite.
Thus, data processing and additivity give
\begin{align}
 \mathbb{D}_{\alpha}(p\|\gamma)
 +\mathbb{D}_{\alpha}(\nu_{\varepsilon}\|\eta_{\varepsilon})
 =
 \mathbb{D}_{\alpha}
 \bigl(p\otimes\nu_{\varepsilon}
 \big\|
 \gamma\otimes\eta_{\varepsilon}\bigr)
 \geq
 \mathbb{D}_{\alpha}
 \bigl(\rho_{\varepsilon}\otimes\nu_{\varepsilon}
 \big\|
 \gamma\otimes\eta_{\varepsilon}\bigr)=
 \mathbb{D}_{\alpha}(\rho_{\varepsilon}\|\gamma)
 +\mathbb{D}_{\alpha}(\nu_{\varepsilon}\|\eta_{\varepsilon}),
\end{align}
and hence
\begin{align}
 \mathbb{D}_{\alpha}(p\|\gamma)
 \geq
 \mathbb{D}_{\alpha}(\rho_{\varepsilon}\|\gamma).
 \label{eq:forward-necessary}
\end{align}

For the reverse ordering,
$\mathbb{D}_{\alpha}(\eta_{\varepsilon}\|\nu_{\varepsilon})$ may be
infinite when $\nu_{\varepsilon}$ is not faithful. Let
$P_{\varepsilon}$ denote the support projection of
$\nu_{\varepsilon}$ and set
$M_{\varepsilon}=I\otimes P_{\varepsilon}$. 
Defining 
\begin{align}
 \eta_{P,\varepsilon}
 =
 \frac{
 P_{\varepsilon}\eta_{\varepsilon}P_{\varepsilon}
 }{
 \Tr[P_{\varepsilon}\eta_{\varepsilon}]
 },
\end{align}
By Lemma~\ref{lem:support-reduction}, the restricted map $ \mathcal{F}_{\varepsilon,M}(X)=M_\varepsilon \mathcal{F}_\varepsilon(X) M_\varepsilon$ is CPTP on $\supp(M_\varepsilon)$. Moreover, it preserves the projected Gibbs state. Hence, 
\begin{align}
 \mathcal{F}_{\varepsilon,M}
 \bigl(p\otimes\nu_{\varepsilon}\bigr)
 =
 \rho_{\varepsilon}\otimes\nu_{\varepsilon},
 \quad 
 \mathcal{F}_{\varepsilon,M}
 \bigl(\gamma\otimes\eta_{P,\varepsilon}\bigr)
 =
 \gamma\otimes\eta_{P,\varepsilon}.
\end{align}
Since the states $\nu_{\varepsilon}$ and $\eta_{P,\varepsilon}$  have coinciding supports, the relative entropies are finite. Therefore,
\begin{align}
 \mathbb{D}_{\alpha}(\gamma\|p)
 +\mathbb{D}_{\alpha}
   (\eta_{P,\varepsilon}\|\nu_{\varepsilon})
 =
 \mathbb{D}_{\alpha}
 \bigl(\gamma\otimes\eta_{P,\varepsilon}
 \big\|
 p\otimes\nu_{\varepsilon}\bigr)
 \geq
 \mathbb{D}_{\alpha}
 \bigl(\gamma\otimes\eta_{P,\varepsilon}
 \big\|
 \rho_{\varepsilon}\otimes\nu_{\varepsilon}\bigr)
 =
 \mathbb{D}_{\alpha}(\gamma\|\rho_{\varepsilon})
 +\mathbb{D}_{\alpha}
   (\eta_{P,\varepsilon}\|\nu_{\varepsilon}),
\end{align}
which yields
\begin{align}
 \mathbb{D}_{\alpha}(\gamma\|p)
 \geq
 \mathbb{D}_{\alpha}(\gamma\|\rho_{\varepsilon}).
 \label{eq:reverse-necessary}
\end{align}

Finally, since $\rho_{\varepsilon}\to\rho$ and the divergences in
$\mathcal{D}$ are continuous on faithful pairs, taking
$\varepsilon\to0$ in
\eqref{eq:forward-necessary} and \eqref{eq:reverse-necessary} gives
\begin{align}
 \mathbb{D}_{\alpha}(p\|\gamma)
 \geq
 \mathbb{D}_{\alpha}(\rho\|\gamma),
 \quad 
 \mathbb{D}_{\alpha}(\gamma\|p)
 \geq
 \mathbb{D}_{\alpha}(\gamma\|\rho).
\end{align}
These are precisely the conditions in
\eqref{eq:gpo-forward-geometric}--\eqref{eq:gpo-reverse-alpha-z}.
\end{proof}

\section{Acknowledgments}
We thank Marco Tomamichel for discussions. R.R. acknowledges financial support from the ERC grant GIFNEQ 101163938.

\section{Statement on the use of artificial intelligence}
R.R. conceived the framework, formulated the research question, conjectured the main result, and developed its applications. ChatGPT 5.6 Sol assisted in writing the manuscript and, under R.R.’s direction and explicit technical guidance, in deriving the optimal preparation map used in the proof of Lemma~\ref{lem: upper bound} and the full-rank approximation constructed in the proof of Lemma~\ref{lem:prepared-coinciding-support}.

\bibliographystyle{ultimate}
\bibliography{library}

\newpage

\appendix

\section{Useful lemmas}
The following lemma states that Schatten norms are contractive under pinching.

\begin{lemma}[Schatten contractivity of pinching]
\label{lem: Schatten contractivity}
Let \(q\ge1\) and let $L,X$ be Hermitian operators. Then,
\begin{align}
    \|\mathcal{P}_L(X)\|_q\le \|X\|_q.
\end{align}
In particular, if \(X\ge0\), then
\begin{align}
    \Tr [\mathcal{P}_L(X)^q]\le \Tr [X^q].
\end{align}
\end{lemma}

\begin{proof}
Let us denote by $m=|\text{spec}(L)|$ and $P_j$ the orthogonal projectors onto the eigenspace of $L$ corresponding to the $j$-th eigenvalue.
Let us define the unitaries
\begin{align}
    U_\ell=\sum_{j=1}^{m}e^{\frac{2\pi i\ell j}{m}}P_j,
    \qquad
    \ell=0,\ldots,m-1.
\end{align}
Then,
\begin{align}
    \mathcal{P}_L(X)
    =
    \frac1m
    \sum_{\ell=0}^{m-1}
    U_\ell XU_\ell^\dagger.
\end{align}
Since the Schatten \(q\)-norm is convex for \(q\ge1\) and unitarily
invariant,
\begin{align}
     \|\mathcal{P}_L(X)\|_q
    \le
    \frac1m
    \sum_{\ell=0}^{m-1}
    \|U_\ell XU_\ell^\dagger\|_q
    =
    \|X\|_q.
\end{align}
If \(X\ge0\), then $\|X\|_q^q=\Tr X^q,$
and the lemma follows.
\end{proof}

In the following lemma, we show that whenever the states have coinciding supports, the optimization in the preparation map can be restricted to preparation distributions having the same support.
\begin{lemma}
\label{lem:prepared-coinciding-support}
Let \(\rho\) and \(\sigma\) be quantum states with coinciding supports. Then, for every \(\alpha\in(0,\infty]\),
\begin{align}
    D_{\alpha}^{\mathbb{P}}(\rho\|\sigma)
    =
    \inf_{\substack{p,q,\mathcal{F}\\
    \mathcal{F}(p)=\rho,\ \mathcal{F}(q)=\sigma\\
    \operatorname{supp}(p)=\operatorname{supp}(q)}}
    D_{\alpha}(p\|q).
    \label{eq:prepared-coinciding-support}
\end{align}
\end{lemma}

\begin{proof}
The left-hand side of
\eqref{eq:prepared-coinciding-support} is no larger than the
right-hand side, since the latter infimum is taken over the smaller
class of preparations whose classical distributions have coinciding
supports. It therefore remains to prove the reverse inequality.

We may restrict the underlying Hilbert space to the common support of
\(\rho\) and \(\sigma\), so that both states are faithful. If
\(\rho=\sigma\), the claim follows from the one-outcome preparation
\(p=q=(1)\). We may therefore assume that \(\rho\neq\sigma\).
Fix a preparation of finite cost,
\begin{align}
    \rho
    =
    \sum_x p_x\tau_x,
    \quad 
    \sigma
    =
    \sum_x q_x\tau_x.
\end{align}
After removing any outcomes for which \(p_x=q_x=0\), we construct the
distributions
\begin{align}
    \widetilde{p}_{\lambda,\delta}
    =
    (1-\lambda)p_\delta\oplus\lambda r,
    \quad 
    \widetilde{q}_{\lambda,\delta}
    =
    (1-\lambda)q_\delta\oplus\lambda s,
    \label{eq:equal-support-preparation-distributions}
\end{align}
together with the preparation ensemble
\begin{align}
\label{eq: ensamble}
    \{\tau_x\}_x
    \cup
    \{A_{\lambda,\delta},B_{\lambda,\delta}\}.
\end{align}
We now show that these distributions have coinciding supports, prepare
\(\rho\) and \(\sigma\) exactly, and recover the cost of the original
preparation as the parameters tend to zero.

Set $\Delta
    =
    \rho-\sigma,$
and choose \(t>0\) sufficiently small that $A
    =
    \rho+t\Delta$ and $
    B
    =
    \sigma-t\Delta$
are faithful states. Define the full-rank binary distributions
\begin{align}
    r
    =
    \frac{1}{1+2t}(1+t,t),
    \quad 
    s
    =
    \frac{1}{1+2t}(t,1+t).
    \label{eq:faithful-anchor-distributions}
\end{align}
A direct calculation gives
\begin{align}
    r_1A+r_2B
    &=
    \frac{(1+t)(\rho+t\Delta)+t(\sigma-t\Delta)}
         {1+2t}
    =
    \rho,
    \\
    s_1A+s_2B
    &=
    \frac{t(\rho+t\Delta)+(1+t)(\sigma-t\Delta)}
         {1+2t}
    =
    \sigma.
\end{align}
For \(\delta\in(0,1)\), define
\begin{align}
    p_\delta
    =
    (1-\delta)p+\delta q,
    \quad 
    q_\delta
    =
    (1-\delta)q+\delta p.
    \label{eq:classical-support-perturbation}
\end{align}
These distributions satisfy $\operatorname{supp}(p_\delta)
    =
    \operatorname{supp}(q_\delta)
    =
    \operatorname{supp}(p)\cup\operatorname{supp}(q)$
and 
\begin{align}
    \sum_x p_{\delta,x}\tau_x
    &=
    (1-\delta)\rho+\delta\sigma
    =
    \rho-\delta\Delta,
    \\
    \sum_x q_{\delta,x}\tau_x
    &=
    (1-\delta)\sigma+\delta\rho
    =
    \sigma+\delta\Delta.
    \label{eq:perturbed-barycentres}
\end{align}
Fix \(\lambda\in(0,1)\) and define
\begin{align}
    A_{\lambda,\delta}
    =
    A
    +(1+2t)\frac{(1-\lambda)\delta}{\lambda}\Delta,
    \quad 
    B_{\lambda,\delta}
    =
    B
    -(1+2t)\frac{(1-\lambda)\delta}{\lambda}\Delta.
    \label{eq:corrected-anchor-states}
\end{align}
For every fixed \(\lambda\), these operators converge to \(A\) and
\(B\), respectively, as \(\delta\to0\). Since \(A\) and \(B\) are
faithful and \(\operatorname{Tr}\Delta=0\), the operators
\(A_{\lambda,\delta}\) and \(B_{\lambda,\delta}\) are states for all
sufficiently small \(\delta>0\). Moreover,
\begin{align}
    r_1A_{\lambda,\delta}+r_2B_{\lambda,\delta}
    &=
    \rho
    +\frac{(1-\lambda)\delta}{\lambda}\Delta,
    \\
    s_1A_{\lambda,\delta}+s_2B_{\lambda,\delta}
    &=
    \sigma
    -\frac{(1-\lambda)\delta}{\lambda}\Delta.
    \label{eq:corrected-anchor-barycentres}
\end{align}
Combining the above relations, we obtain
\begin{align}
    &(1-\lambda)\sum_x p_{\delta,x}\tau_x
    +\lambda
    \bigl(
        r_1A_{\lambda,\delta}
        +r_2B_{\lambda,\delta}
    \bigr)
    =
    \rho,
    \\
    &(1-\lambda)\sum_x q_{\delta,x}\tau_x
    +\lambda
    \bigl(
        s_1A_{\lambda,\delta}
        +s_2B_{\lambda,\delta}
    \bigr)
    =
    \sigma.
\end{align}
Thus, the distributions in
\eqref{eq:equal-support-preparation-distributions}, together with the
ensemble in~\eqref{eq: ensamble}, prepare \(\rho\) and \(\sigma\) exactly. They also
have coinciding supports because \(p_\delta\) and \(q_\delta\) have
coinciding supports and \(r\) and \(s\) both have full support.

For \(\alpha\neq1\), additivity over direct sums and joint homogeneity
give
\begin{align}
    Q_\alpha\bigl(
        \widetilde{p}_{\lambda,\delta}
        \big\|
        \widetilde{q}_{\lambda,\delta}
    \bigr)
    &=
    (1-\lambda)Q_\alpha(p_\delta\|q_\delta)
    +
    \lambda Q_\alpha(r\|s).
    \label{eq:equal-support-Q-cost}
\end{align}
For \(\alpha\in (0,1)\), continuity yields
\begin{align}
    \lim_{\delta\to0}
    Q_\alpha(p_\delta\|q_\delta)
    &=
    Q_\alpha(p\|q).
\end{align}
For \(\alpha>1\), finite cost implies $\operatorname{supp}(p)
    \subseteq
    \operatorname{supp}(q)$.
After the common-zero outcomes have been removed, \(q\) therefore has
full support, and the same continuity argument applies.
At \(\alpha=1\), the direct-sum decomposition instead gives
\begin{align}
    D_1\bigl(
        \widetilde{p}_{\lambda,\delta}
        \big\|
        \widetilde{q}_{\lambda,\delta}
    \bigr)
    &=
    (1-\lambda)D_1(p_\delta\|q_\delta)
    +
    \lambda D_1(r\|s),
    \label{eq:equal-support-relative-entropy-cost}
\end{align}
with $\lim_{\delta\to0}
    D_1(p_\delta\|q_\delta)
    =
    D_1(p\|q)$.
Consequently, for every finite \(\alpha>0\),
\begin{align}
    \lim_{\lambda\to0}
    \lim_{\delta\to0}
    D_\alpha\bigl(
        \widetilde{p}_{\lambda,\delta}
        \big\|
        \widetilde{q}_{\lambda,\delta}
    \bigr)
    &=
    D_\alpha(p\|q).
    \label{eq:equal-support-cost-convergence}
\end{align}
Thus, every finite-cost preparation can be approximated arbitrarily
well by preparations whose classical distributions have coinciding
supports. Taking the infimum proves the reverse inequality in
\eqref{eq:prepared-coinciding-support}, and hence the claim.

It remains to consider \(\alpha=\infty\). Set $D_{\max}(\rho\|\sigma)
    =
    \log c$ where $c
    =
    \min\{c'>0:\rho\leq c'\sigma\}$.
For any preparation with distributions $p$ and $q$ and ensemble $\{\tau_x\}_x$, let
\begin{align}
    M
    &=
    \max_x\frac{p_x}{q_x}
    =
    \exp\bigl(D_{\max}(p\|q)\bigr).
\end{align}
Then,
\begin{align}
    \rho
    =
    \sum_xp_x\tau_x
    &\leq
    M\sum_xq_x\tau_x
    =
    M\sigma,
\end{align}
and hence $D_{\max}(p\|q)
    \geq
    D_{\max}(\rho\|\sigma)$.
Conversely, \(c>1\) because \(\rho\neq\sigma\), and
\begin{align}
    B_\infty
    &=
    \frac{c\sigma-\rho}{c-1}
    \label{eq:max-anchor-state}
\end{align}
is a state. With \(A=\rho+t(\rho-\sigma)\), define $L
    =
    1+t+\frac{1}{c-1}$ and 
\begin{align}
    p_\infty
    =
    \frac{1}{L}
    \left(
    1+\frac{1}{c-1},t
    \right),
    \quad 
    q_\infty
    =
    \frac{1}{L}
    \left(
    \frac{1}{c-1},1+t
    \right).
    \label{eq:max-anchor-distributions}
\end{align}
These distributions have full support and satisfy
\begin{align}
    \rho
&=
    (p_\infty)_1A+(p_\infty)_2B_\infty,
    \quad 
    \sigma
    =
    (q_\infty)_1A+(q_\infty)_2B_\infty.
\end{align}
Furthermore,
\begin{align}
D_{\max}(p_\infty\|q_\infty)
    &=
    \log\max\left\{
    c,\frac{t}{1+t}
    \right\}
    =
    \log c
    =
    D_{\max}(\rho\|\sigma).
    \label{eq:max-anchor-cost}
\end{align}
This proves \eqref{eq:prepared-coinciding-support} also for \(\alpha=\infty\).
\end{proof}

\section{Support-reduction lemma}
In this appendix, we establish a support-reduction lemma needed to rigorously prove the converse direction of Theorem~\ref{thm:catalytic-gibbs-preserving} for the divergences with reversed arguments. The only subtlety arises when the support of the catalyst state is strictly contained in that of its Gibbs state, in which case the relative entropies may be infinite. We resolve this issue by restricting the channel to the support of the catalyst state and replacing the catalyst Gibbs state with its normalized projection onto the same subspace. The restricted map is CPTP and preserves the corresponding total Gibbs state. The relevant catalyst relative entropies are then finite, so the standard argument based on additivity and the data-processing inequality yields the required converse inequalities.

We first introduce some notation and terminology used in the proof. Let
\(\mathcal{F}:\mathcal{L}(\mathcal{H})\to\mathcal{L}(\mathcal{H})\)
be a quantum channel with Kraus representation
\begin{align}
    \mathcal{F}(X)
    =\sum_j K_j X K_j^\dagger,
    \qquad
    \sum_jK_j^\dagger K_j=I.
    \label{eq:kraus-representation}
\end{align}
We denote by \(\mathcal{F}^\dagger\) its adjoint with respect to the
Hilbert--Schmidt inner product, defined by
\begin{align}
    \Tr\big[Y^\dagger\mathcal{F}(X)\big]
    =\Tr\big[\bigl(\mathcal{F}^\dagger(Y)\bigr)^\dagger X\big],
    \qquad
    \mathcal{F}^\dagger(Y)=\sum_jK_j^\dagger YK_j.
    \label{eq:adjoint-definition}
\end{align}

If \(M\) is an orthogonal projection, we write
\(\supp(M)=M\mathcal{H}\) and set \(M^\perp=I-M\). We say that
\(\supp(M)\) reduces \(\mathcal{F}\) if
\begin{align}
    M^\perp K_jM=0
    \quad\text{and}\quad
    MK_jM^\perp=0
    \qquad\forall j.
    \label{eq:reducing-subspace}
\end{align}
In this case, every Kraus operator is block diagonal with respect to
\(\mathcal{H}=\supp(M)\oplus\supp(M^\perp)\), so the channel cannot
transfer support between the two subspaces.

Assumption~\eqref{eq:image-supported-in-M}, together with the fact that
\(\mu\) is faithful on \(\supp(M)\), first shows that no component can
flow from \(\supp(M)\) to \(\supp(M^\perp)\), and hence that
\(\supp(M)\) is invariant under \(\mathcal F\). The existence of a faithful fixed state then
rules out flow in the opposite direction and establishes reducibility.
This stronger property is needed to ensure that the projected fixed
state remains invariant under the restricted channel.

We now state the lemma.

\begin{lemma}
\label{lem:support-reduction}
Let \(\mathcal{F}:\mathcal{L}(\mathcal{H})\to
\mathcal{L}(\mathcal{H})\) be a quantum channel admitting a faithful
fixed state \(\omega>0\), so that
\begin{align}
    \mathcal{F}(\omega)=\omega.
    \label{eq:faithful-fixed-state}
\end{align}
Let \(0\neq\mu\geq0\), let \(M\) be its support projection, and suppose
that
\begin{align}
    \supp\!\left(\mathcal{F}(\mu)\right)\subseteq \supp(M).
    \label{eq:image-supported-in-M}
\end{align}
Then the following statements hold.
\begin{enumerate}
    \item The subspace \(\supp(M)\) reduces \(\mathcal{F}\).
    \label{item: first}
    \item The restricted map \(\mathcal{F}_M:\mathcal{L}(\supp(M))\to\mathcal{L}(\supp(M))\), defined by \(\mathcal{F}_M(X)=M\mathcal{F}(X)M\), is completely positive and trace preserving.
    \label{item: second}
    \item The normalized projection
    \begin{align}
        \omega_M
        =\frac{M\omega M}{\Tr[M\omega]}
        \label{eq:compressed-fixed-state}
    \end{align}
    is faithful on \(\supp(M)\) and satisfies
    \begin{align}
        \mathcal{F}_M(\omega_M)=\omega_M.
        \label{eq:compressed-fixed-state-invariant}
    \end{align}
    \label{item: third}
\end{enumerate}
\end{lemma}

\begin{proof}
Set \(M^\perp=I-M\), and fix the Kraus representation
\eqref{eq:kraus-representation}. We prove the three statements in order.

\medskip
\noindent
We first prove statement~\ref{item: first}. We begin by showing that \(\supp(M)\) is
invariant. Since \(M\) is the support projection of \(\mu\), we have
\(\mu=M\mu M\), and the restriction of \(\mu\) to \(\supp(M)\) is
strictly positive. Assumption~\eqref{eq:image-supported-in-M} gives
\begin{align}
    0
    =M^\perp \mathcal{F}(\mu)M^\perp
    =\sum_jM^\perp K_jM\mu MK_j^\dagger M^\perp.
    \label{eq:no-flow-out-sum}
\end{align}
Every summand on the right-hand side is positive semidefinite. A sum of
positive semidefinite operators can vanish only if every summand
vanishes. Consequently,
\begin{align}
    M^\perp K_jM\mu MK_j^\dagger M^\perp=0
    \qquad\text{for every }j.
\end{align}
Equivalently,
\begin{align}
    M^\perp K_jM\mu^{1/2}=0
    \qquad\text{for every }j.
\end{align}
The operator \(\mu^{1/2}\) is invertible on \(\supp(M)\), and hence
\begin{align}
    M^\perp K_jM=0
    \qquad\text{for every }j.
    \label{eq:no-flow-out}
\end{align}
Thus, every state supported on \(\supp(M)\) is mapped to a state
supported on the same subspace. This proves that \(\supp(M)\) is
invariant.

We next show that no component can flow from \(\supp(M^\perp)\) into
\(\supp(M)\). We now verify that
invariance implies
\(\mathcal{F}^\dagger(M)\geq M\). First,
\begin{align}
    M\mathcal{F}^\dagger(M)M
    =\sum_jMK_j^\dagger MK_jM
    =\sum_jMK_j^\dagger K_jM
    =M,
    \label{eq:MM-block}
\end{align}
where the second equality uses \eqref{eq:no-flow-out}, and the last one
uses \(\sum_jK_j^\dagger K_j=I\). Similarly, trace preservation and
\eqref{eq:no-flow-out} give
\begin{align}
    M\mathcal{F}^\dagger(M)M^\perp
    =\sum_jMK_j^\dagger MK_jM^\perp
    =\sum_jMK_j^\dagger K_jM^\perp
    =MIM^\perp
    =0.
    \label{eq:MMperp-block}
\end{align}
Taking the adjoint shows that
\(M^\perp\mathcal{F}^\dagger(M)M=0\), while
\begin{align}
    M^\perp\mathcal{F}^\dagger(M)M^\perp
    =\sum_jM^\perp K_j^\dagger MK_jM^\perp
    =\sum_j(MK_jM^\perp)^\dagger(MK_jM^\perp)
    \geq0.
    \label{eq:MperpMperp-block}
\end{align}
Combining \eqref{eq:MM-block}--\eqref{eq:MperpMperp-block}, we obtain
\begin{align}
    \mathcal{F}^\dagger(M)-M
    =\sum_j(MK_jM^\perp)^\dagger(MK_jM^\perp)
    \geq0,
    \label{eq:subharmonicity-explicit}
\end{align}
where the operator on the right-hand side is supported on
\(\supp(M^\perp)\).

We now evaluate this positive operator in the faithful fixed state
\(\omega\). Using the defining relation for the adjoint and
\eqref{eq:faithful-fixed-state},
\begin{align}
    \Tr\big[\omega\bigl(\mathcal{F}^\dagger(M)-M\bigr)\big]
    =\Tr[\mathcal{F}(\omega)M]
    -\Tr[\omega M]
    =0.
    \label{eq:stationary-mass-balance}
\end{align}
Let \(A=\mathcal{F}^\dagger(M)-M\geq0\). Since \(\omega>0\) and the
Hilbert space is finite-dimensional, there exists \(m>0\) such that
\(\omega\geq mI\). Therefore,
\begin{align}
    0=\Tr[\omega A]\geq m\Tr[A]\geq0.
\end{align}
It follows that \(\Tr[A]=0\), and the positivity of \(A\) then implies
\(A=0\). Hence,
\begin{align}
    \mathcal{F}^\dagger(M)=M.
    \label{eq:harmonic-projection}
\end{align}
Taking the \(M^\perp\)-block in \eqref{eq:harmonic-projection} and using
\eqref{eq:MperpMperp-block}, we obtain
\begin{align}
    0
    =M^\perp\mathcal{F}^\dagger(M)M^\perp
    =\sum_j(MK_jM^\perp)^\dagger(MK_jM^\perp).
\end{align}
Every summand is positive semidefinite, so
\begin{align}
    MK_jM^\perp=0
    \qquad\text{for every }j.
    \label{eq:no-flow-in}
\end{align}
Together, \eqref{eq:no-flow-out} and \eqref{eq:no-flow-in} show that
\(\supp(M)\) reduces \(\mathcal{F}\), proving statement~\ref{item: first}.

\medskip
\noindent
We now prove statement~\ref{item: second}. Define
\begin{align}
    A_j=MK_jM.
\end{align}
For \(X\in\mathcal{L}(\supp(M))\), identified with an operator satisfying
\(X=MXM\), we have
\begin{align}
    \mathcal{F}_M(X)=\sum_jA_jXA_j^\dagger.
    \label{eq:restricted-kraus}
\end{align}
Thus, \(\mathcal{F}_M\) is completely positive. Moreover,
\begin{align}
    \sum_jA_j^\dagger A_j
    =\sum_jMK_j^\dagger MK_jM
    =\sum_jMK_j^\dagger K_jM
    =M.
    \label{eq:restricted-trace-preserving}
\end{align}
The projection \(M\) is the identity operator on the reduced Hilbert
space \(\supp(M)\). Equation~\eqref{eq:restricted-trace-preserving} is
therefore precisely the trace-preservation condition for
\(\mathcal{F}_M\).

\medskip
\noindent
We finally prove statement~\ref{item: third}. The reducing relations imply that every
\(K_j\) is block diagonal. Consequently, the \(M\)-block of
\(\mathcal{F}(\omega)\) depends only on the \(M\)-block of \(\omega\).
Explicitly,
\begin{align}
    \mathcal{F}_M(M\omega M)
    =\sum_jMK_jM\,(M\omega M)\,MK_j^\dagger M
    =M\mathcal{F}(\omega)M
    =M\omega M.
    \label{eq:compressed-state-fixed}
\end{align}
Because \(\omega>0\), the compression \(M\omega M\) is strictly
positive on \(\supp(M)\), and its trace is nonzero. Dividing
\eqref{eq:compressed-state-fixed} by \(\Tr[M\omega]\) proves
\eqref{eq:compressed-fixed-state-invariant}.
\end{proof}

Notice that \(X\mapsto M\mathcal{F}(X)M\) need not be trace preserving
when viewed as a map on the original full algebra
\(\mathcal{L}(\mathcal{H})\). It is a quantum channel because its domain
and codomain are restricted to \(\mathcal{L}(\supp(M))\), on which
\(M\) is the identity operator.

\end{document}